\documentclass[lettersize,onecolumn]{IEEEtran}
\usepackage{amsmath,amsfonts}
\usepackage{algorithmic}
\usepackage{algorithm}
\usepackage{array}
\usepackage[caption=false,font=normalsize,labelfont=sf,textfont=sf]{subfig}
\usepackage{textcomp}
\usepackage{stfloats}
\usepackage{url}
\usepackage{verbatim}
\usepackage{graphicx}
\usepackage{cite}
\usepackage[margin=1in]{geometry}
\usepackage{amsmath,amssymb,amsthm}
\usepackage{mathtools}
\usepackage{hyperref}
\usepackage{cite}
\usepackage{amsmath,amssymb,booktabs,xcolor}
\usepackage{float}

\usepackage{amsmath,amssymb,booktabs,xcolor}
\newcommand{\Ftwo}{\mathbb{F}_2}
\newcommand{\SG}[2]{\mathrm{SmallGroup}(#1,#2)}
\newcommand{\Core}{\mathrm{Core}}

\newtheorem{theorem}{Theorem}[section]

\newtheorem{lemma}[theorem]{Lemma}
\newtheorem{proposition}[theorem]{Proposition}
\newtheorem{definition}[theorem]{Definition}

\newtheorem{example}[theorem]{Example}
\newtheorem{remark}[theorem]{Remark}
\begin{document}

\title{Coset Lifted Product Codes: Constructing Quantum LDPC Codes from Non-Normal Subgroup Actions}

\author{
\IEEEauthorblockN{
Rajendra Prasad Rajpurohit$^{1}$,
V. Lalitha$^{2}$,
Shayan Srinivasa Garani$^{1}$
} \\
\IEEEauthorblockA{
$^{1}$\textit{Department of Electronic Systems Engineering},
\textit{Indian Institute of Science}, Bengaluru, India\\
$^{2}$\textit{Signal Processing and Communications Research Center},
\textit{International Institute of Information Technology Hyderabad},
India\\
rajendrarajpurohit20@gmail.com,
lalitha.v@iiit.ac.in,
shayangs@iisc.ac.in
}
}

% The paper headers

\maketitle

\begin{abstract}
We introduce coset lifted product (CLP) codes, a family of quantum LDPC codes that combine the protograph structure of lifted product codes with the coset action of Aydin, Tamo, and Barg \cite{AydinTamoBarg2026}. Given a finite group $G$, a subgroup $H \le G$ that need not be normal, a protograph $A$ over $\mathbb F_2[G]$, and a protograph $B$ over $\mathbb F_2[N_G(H)]$, the two protographs act on the cosets $G/H$ from the left and from the right to produce quantum CSS codes. Further, for groups of odd order, we give an explicit formula for the dimension as a sum over the irreducible representations $\rho$ of $G$ of hypergraph product dimensions. A search over $14,843$ pairs $(G,H)$ with $H$ restricted to be non-normal and $120 \le n \le 400$ yields certified codes that exceed the best published value of $Q=kd^2/n$ at check weights $6,7$ and $9$, i.e., we obtain $[[360,20,20]]$ at $w=6$, a $[[360,24,20]]$ code at $w=7$, and codes $[[252,16,18]]$ and $[[288,16,20]]$ at $w=9$. Our search also produce high-distance codes with parameters $[[168,10,14]], [[240,8,20]]$ at $w=6$ and $[[216,8,18]]$ at $w=7$.
\end{abstract}


\section{Introduction}
\label{sec:intro}

The idea that quantum information could be protected from noise was not obvious until Shor showed, in 1995, that a quantum state could be encoded across several physical qubits, and errors detected and reversed without the state being directly measured, which could  lead to the collapsing the encoded state~\cite{Shor1995} otherwise. This opened the direction of the  general theory of quantum error-correcting codes for encoding of $k$ logical qubits into $n>k$ physical qubits, designed to detect and correct errors while leaving the protected information untouched~\cite{NielsenChuang2000}. A particularly convenient family of such codes was introduced independently by Calderbank and Shor, and by Steane, in 1996, known as Calderbank-Shor-Steane (CSS) construction~\cite{CalderbankShor1996,Steane1996}. The CSS construction builds a quantum code using two classical linear codes. If two linear codes having parity-check matrices $H_X$ and $H_Z$ satisfy $H_X H_Z^{\mathsf{T}} = 0$, the resulting $X$-type and $Z$-type Pauli stabilizers commute to produce a quantum error-correcting code for correcting independent bit flips and phase flips. This turned the design of a large class of quantum codes into a purely combinatorial question, namely finding two classical codes whose check matrices are mutually orthogonal.

An important application of this idea followed shortly afterward. Kitaev proposed the toric code, using local stabilizer checks on a qubit lattice on a torus~\cite{Kitaev2003}. Bravyi and Kitaev subsequently adapted the construction to a planar geometry with boundaries~\cite{BravyiKitaev1998}, giving what is now known as the surface code, and Dennis et al.\ established its viability as a quantum memory~\cite{DennisKitaevLandahlPreskill2002}. Later, Fowler et al.~\cite{Fowler2012} laid out in detail why this code family is useful for near-term hardware, showing that its stabilizers are local and two-dimensional, and it tolerates a comparatively high rate of circuit-level noise. What this line of work did not solve was the question of rate. Because the surface code needs $n = \Theta(d^2)$ physical qubits to protect a single logical qubit to distance $d$, its rate $k/n$ vanishes as the code length scaled up, and fault-tolerant architecture built purely from surface codes inherits this overhead.

This shortcoming motivated the search for quantum low-density parity-check (LDPC) codes. The term is borrowed from Gallager, who introduced classical LDPC codes decades earlier as codes whose parity-check matrices are sparse, with row and column weights bounded independently of the code length~\cite{Gallager1962}. A quantum LDPC code is a family of stabilizer codes whose checks are sparse in the same sense, i.e., each check acts on a bounded number of qubits, and each qubit enters a bounded number of checks. For a CSS code, this means that the row and column weights of both $H_X$ and $H_Z$ are bounded. An early systematic study of such codes is due to MacKay, Mitchison, and McFadden~\cite{MacKayMitchisonMcFadden2004}. Interest in the class grew once Gottesman showed that a quantum LDPC family with a constant rate, together with an efficient decoder that suppresses errors below a constant noise threshold, allows fault-tolerant quantum computation with constant space overhead, as opposed to the polylogarithmic overhead of concatenated schemes~\cite{Gottesman2014}. A family with constant rate $k/n$ and linear distance $d = \Theta(n)$ is called \emph{good}. Whether good quantum LDPC codes exist at all remained open for nearly two decades.

Obtaining constant-rate quantum LDPC codes with large distance proved much harder than the classical theory of LDPC and expander codes~\cite{SipserSpielman1996} would suggest. The hypergraph product of Tillich and Z\'emor~\cite{TillichZemor2009,TillichZemor2014} turns any two classical LDPC codes into a quantum LDPC code with constant rate, but its distance is $d = \Theta(\sqrt n)$, and for over a decade no quantum LDPC construction of any rate exceeded $d = O(\sqrt n\,\mathrm{polylog}\, n)$.

Later, Panteleev and Kalachev~\cite{PanteleevKalachev2021Degenerate,PanteleevKalachev2022LinearDistance} introduced the \emph{lifted product}, a hypergraph product taken over a commutative subalgebra $R \subseteq M_\ell(\mathbb F_2)$ instead of over $\mathbb F_2$, with the algebra of $\ell\times\ell$ circulant matrices, i.e.\ $R = \mathbb F_2[\mathbb Z_\ell]$, as the case used in practice. Applied to a random cyclic cover of an expander graph, this gave quantum LDPC codes with dimension $\Theta(\log n)$ and \emph{almost linear} distance $\Theta(n/\log n)$. The balanced product of Breuckmann and Eberhardt~\cite{BreuckmannEberhardt2021Balanced}, introduced at the same time, obtains a quantum code by taking the quotient of a hypergraph product by a common group action on both factors, when the group is abelian and acts freely, it coincides with a lifted product. Neither family has a constant rate.

Constant rate together with linear distance took a different route. Panteleev and Kalachev~\cite{PanteleevKalachev2022Good} took the lifted product over the group algebra of a \emph{non-abelian} group and proved that, under expansion conditions on the local codes, the resulting quantum LDPC codes are asymptotically good. This settled the qLDPC conjecture. Quantum Tanner codes~\cite{LeverrierZemor2022Tanner} reach the same parameters by a closely related construction. For a survey of the field see~\cite{BreuckmannEberhardt2021Review}.

At the block lengths relevant to current hardware i.e., over short code lengths, attention has centered on explicit lifted product codes with $1\times 1$ protographs over the regular representation of a finite group, known as two-block codes. Generalized bicycle codes~\cite{KovalevPryadko2013} use a cyclic group. Two-block group algebra (2BGA) codes~\cite{LinPryadko2024} allow any finite group $G$, for $a, b \in \mathbb F_2[G]$ acting through the left and right regular representations $L$ and $R$, the code has $H_X = (L(a)\,|\,R(b))$ and $H_Z = (R(b)^{\mathsf T}\,|\,L(a)^{\mathsf T})$. Bivariate bicycle (BB) codes~\cite{Bravyi2024} are the case $G = \mathbb Z_\ell \times \mathbb Z_{\ell'}$ with $a$ and $b$ each a sum of three monomials. Bravyi et al.\ report a circuit-level threshold close to that of the surface code at an order of magnitude fewer physical qubits. In each of these constructions, the group acts on itself, so the block algebra is the full regular representation.

Two further directions are relevant here. One grows a base code into a family by graph covers~\cite{SymonsRajputBrowne2025,Guemard2024} or by group extensions~\cite{HirasakiLee2026}. In the latter the kernel of the extension is normal, so the lifted code is again a lifted product over the regular representation of a group. The other direction is by changing the group actions over the bi-partitions of the $H_X$ and $H_Z$ matrices as done by Aydin, Tamo, and Barg~\cite{AydinTamoBarg2026}. They let $\mathbb F_2[G]$ act on the coset space $G/H$ of a subgroup $H$ that need not be normal, allowing for a much expanded space of possible codes, and they take the second block from a right action of the normalizer $N_G(H)$ on the same cosets. The two actions commute, which gives CSS orthogonality, and the resulting two-block codes are 2BGA codes when $H$ is normal and lie outside that family otherwise. Their search over pairs $(G,H)$ yields several new codes, including codes of check weight $6,8$ and $10$. The same non-normal actions appear in the balanced product setting~\cite{BreuckmannEberhardt2021Balanced}, where Qian and Li~\cite{qian2026multi} recently ran an automated search over coset-orbit balanced-product codes with protograph entries given by sums of double cosets, retaining a candidate only after verifying $H_XH_Z^{\mathsf T}=0$ numerically.

In this paper, we combine the two ingredients i.e., based on the protographs and the group actions. Given a finite group $G$, a subgroup $H \le G$ with normalizer $N = N_G(H)$, a protograph $A$ over $\mathbb F_2[G]$, and a protograph $B$ over $\mathbb F_2[N]$, we let $A$ act on the coset space $G/H$ from the left and $B$ from the right, and form the lifted product of the resulting block matrices. We call the resulting codes the coset lifted product codes and denote them as $\mathrm{CLP}(A,B;G,H)$. We show that the code satisfies the CSS condition for every $G$ and subgroup $H$, whether normal or not. For $|G|$ odd, we derive formula for dimension of these codes. Further a search over $14,843$ pairs $(G,H)$ with $120 \le n \le 400$ yields certified codes exceeding the best published value of $Q = kd^2/n$ at $w = 6$, a $[[360,20,20]]$ code with $Q = 22.2$, at $w=7$, a $[[360,24,20]]$ code with $Q=26.67$ , and at $w = 9$, a $[[288,16,20]]$ code with $Q = 22.22$, together with candidates at $w \in \{7,9,10\}$ whose distance upper bounds give $Q$ up to $40.5$. All codes are printed in full and are reproducible from the data given.

The remainder of the paper is organized as follows. In Section~\ref{sec:prelims} we discuss the basics of quantum CSS codes, group algebras and constructing quantum codes from lifted product codes based on two-block group algebra codes. In  Section~\ref{sec:construction} we construct coset lifted product CSS codes and obtain the dimension of the codes based on their irreducible representations. In Section~\ref{sec:search} we report the search of such codes along with concrete examples showing the code design procedure to motivate proposed code constructions.

\section{Preliminaries}
\label{sec:prelims}
Throughout, $\ker M$ and $\operatorname{row}(M)$ denote the kernel and the row space of matrix $M$ over $\mathbb F_2$. For a field $\mathbb F$, $M_{p \times q}(\mathbb F)$ denotes the set of $p \times q$ matrices over $\mathbb F$, and $M_p(\mathbb F)$ the square ones.  

\subsection{Classical and quantum codes}
\label{sec:prelim-codes}

 A binary linear code of length $n$ is a subspace $C \subseteq \mathbb F_2^n$. A matrix $H \in \mathbb F_2^{r \times n}$ with $C = \ker H$ is called a \emph{parity-check matrix} of the code $C$. The parameters of $C$ are written $[n,k,d]$, where $k = \dim C = n - \operatorname{rank} H$, and the minimum distance $d = \min\{\operatorname{wt}(v) \mid v \in C,\ v \ne 0\}$, where $\operatorname{wt}(v)$ is the Hamming weight of $v\in C$.    

A family of codes is called a low-density parity-check (LDPC) codes if the row and column weights of its parity-check matrices are bounded by a constant independent of $n$, so that every check involves a bounded number of bits and every bit takes part in a bounded number of checks~\cite{Gallager1962,MacWilliamsSloane1977}.    
For a vector $v \in \mathbb F_2^n$ let $X(v) = \bigotimes_{i} X^{v_i}$ and $Z(v) = \bigotimes_{i} Z^{v_i}$ denote the corresponding $n$-qubit Pauli operators. Two such operators satisfy $X(u)Z(v) = (-1)^{\langle u, v\rangle} Z(v)X(u)$, so $X(u)$ and $Z(v)$ commute exactly when $\langle u, v\rangle = 0$ over $\mathbb F_2$~\cite{Gottesman1997,NielsenChuang2000}.  

 \begin{definition}[CSS code] \label{def:css} Let $H_X \in \mathbb F_2^{r_X \times n}$ and $H_Z \in \mathbb F_2^{r_Z \times n}$ satisfy $H_X H_Z^{\mathsf T} = 0$. The CSS code $\mathrm{CSS}(H_X, H_Z)$~\cite{CalderbankShor1996,Steane1996} is the stabilizer code on $n$ qubits whose stabilizer group is generated by the operators $X(h)$ for the rows $h$ of $H_X$ and $Z(h)$ for the rows $h$ of $H_Z$. The condition $H_X H_Z^{\mathsf T} = 0$ ensures that all generators commute. The code has parameters $[[n,k,d]]$ with \begin{equation}     k = n - \operatorname{rank} H_X - \operatorname{rank} H_Z,     \qquad     d = \min(d_X, d_Z), \end{equation} where \begin{equation}     d_X = \min\big\{\, \operatorname{wt}(v) \ \big|\ v \in \ker H_Z \setminus \operatorname{row}(H_X) \,\big\},     \qquad     d_Z = \min\big\{\, \operatorname{wt}(v) \ \big|\ v \in \ker H_X \setminus \operatorname{row}(H_Z) \,\big\}. \end{equation} The vectors in $\ker H_Z$ correspond to $X$-type operators that commute with every $Z$ stabilizer, those in $\operatorname{row}(H_X)$ are the $X$ stabilizers themselves, and the difference set consists of the nontrivial logical $X$ operators. The roles are exchanged for $Z$. In the language of classical codes, $\mathrm{CSS}(H_X, H_Z)$ is the code $\mathrm{CSS}(C_1, C_2)$ with $C_1 = \ker H_Z$ and $C_2 = \operatorname{row}(H_X) \subseteq C_1$, and $k = \dim C_1 - \dim C_2$. The code is quantum LDPC if $H_X$ and $H_Z$ are LDPC. \end{definition}

\subsection{Group actions and group algebras}
\label{sec:prelim-groups}
In this subsection, we recall definitions of cosets, group actions on them, and the group algebra notation used in the construction.

\begin{definition}[Cosets, normalizer, core]
\label{def:subgroup}
Let $G$ be a finite group and $H \le G$ a subgroup. The set of left cosets is $\Omega = G/H = \{\, gH \mid g \in G \,\}$, and we write $m = |\Omega| = |G|/|H|$. The subgroup $H$ is \emph{normal} in $G$, written $H \trianglelefteq G$, if $gHg^{-1} = H$ for every $g \in G$. The
\emph{normalizer} of $H$ in $G$ is
\begin{equation}
    N_G(H) = \{\, g \in G \mid gHg^{-1} = H \,\},
\end{equation}
the largest subgroup of $G$ in which $H$ is normal, and we abbreviate it to $N$ when $G$ and $H$ are fixed. The \emph{core} of $H$ in $G$ is
\begin{equation}
    \operatorname{Core}_G(H) = \bigcap_{g \in G} gHg^{-1},
\end{equation}
the largest normal subgroup of $G$ contained in $H$. Thus $H \trianglelefteq G$ holds if and only if $N = G$, and if and only if $\operatorname{Core}_G(H) = H$~\cite{DummitFoote2004}.
\end{definition}

\begin{definition}[Permutation matrices]
\label{def:permrep}
Let $\Omega = \{\omega_1, \dots, \omega_m\}$ be an ordered set. To a permutation $\sigma$ of $\Omega$ we associate the $m \times m$ binary matrix with entries
\begin{equation}\label{eq:matconv}
    \sigma_{ij} = 1 \ \text{ if } \sigma(\omega_j) = \omega_i,
    \qquad \sigma_{ij} = 0 \ \text{ otherwise},
\end{equation}
and use the same symbol $\sigma$ for both. Thus $\sigma e_j = e_i$ when $\sigma$ sends $\omega_j$ to $\omega_i$, composing permutations corresponds to multiplying their matrices, and $\sigma^{\mathsf T} = \sigma^{-1}$ for every permutation matrix.
\end{definition}

\begin{definition}[Left and right coset actions]
\label{def:coset-action}
For $g \in G$ and $n \in N$ define permutations $L(g)$ and $R(n)$ of $\Omega = G/H$ by
\begin{equation}\label{eq:LR}
    L(g)(xH) = gxH, \qquad R(n)(xH) = xnH, \qquad x \in G .
\end{equation}
% We use the same symbols for the permutation matrices of Definition~\ref{def:permrep}.
\end{definition}

\begin{proposition}[\cite{AydinTamoBarg2026}, Props.~II.1 and II.3]
\label{prop:LR}
\begin{enumerate}
    \item $R(n)$ is well defined, since $n^{-1}Hn = H$ gives $xhnH = xn(n^{-1}hn)H = xnH$ for all $h \in H$.
    \item $L(gg') = L(g)L(g')$ and $R(nn') = R(n')R(n)$ for all     $g, g' \in G$ and $n, n' \in N$.
    \item $\ker L = \operatorname{Core}_G(H)$ and $\ker R = H$. Hence     $L$ takes $|G|/|\operatorname{Core}_G(H)|$ distinct values and $R$ takes $|N|/|H|$ distinct values.
\end{enumerate}
\end{proposition}

\begin{proof}
Statement 2 follows from~\eqref{eq:LR}, since $L(gg')(xH) = gg'xH = L(g)L(g')(xH)$ and $R(nn')(xH) = xnn'H = R(n')R(n)(xH)$. For 3, $L(g) = \mathrm{id}$ (identity permutation) if and only if $x^{-1}gx \in H$ for all $x \in G$, that is $g \in \bigcap_x xHx^{-1}$, and $R(n) = \mathrm{id}$ if and only if
$xnH = xH$ for all $x$, which for $x = e$ gives $n \in H$, and conversely $n \in H$ gives $xnH = xH$.
\end{proof}

\begin{definition}[Group algebra]
\label{def:groupalg}
For a finite group $G$, the group algebra $\mathbb F_2[G]$ is the $\mathbb F_2$-vector space with basis $G$, whose elements are the formal
sums $\alpha = \sum_{x \in G} a_x\, x$ with $a_x \in \mathbb F_2$, and whose multiplication is
\begin{equation}
    \Big(\sum_{x \in G} a_x\, x\Big)\Big(\sum_{y \in G} b_y\, y\Big)
    = \sum_{x, y \in G} a_x b_y\, xy .
\end{equation}
\end{definition}
We extend the action $L$ and $R$ of Definition \ref{def:coset-action} to the group algebras by linearity, $L(\alpha) = \sum_{x} a_x L(x)$ for $\alpha \in \mathbb F_2[G]$ and
$R(\beta) = \sum_{x} b_x R(x)$ for $\beta \in \mathbb F_2[N]$, and apply them entrywise to matrices with entries in these algebras.

\subsection{Group algebra quantum LDPC codes}
\label{sec:prelim-lp}

This subsection recalls the two-block codes and the lifted product.

\begin{definition}[Two-block group algebra codes]
\label{def:2bga}
Let $G$ be a finite group and let $L_G$ and $R_G$ be the left and right regular representations of $G$, 
% that is, the permutations of $G$ given
% by $L_G(g)(x) = gx$ and $R_G(g)(x) = xg$, which are the coset actions of
% Definition~\ref{def:coset-action} for $H = \{e\}$, so that
% $N = G$ and $L_G(g)$ commutes with $R_G(g')$ for all $g, g' \in G$. 
For $a, b \in \mathbb F_2[G]$ the two-block group algebra (2BGA) code~\cite{LinPryadko2024} is the CSS code with
\begin{equation}
    H_X = \big(L_G(a) \,\big|\, R_G(b)\big), \qquad
    H_Z = \big(R_G(b)^{\mathsf T} \,\big|\, L_G(a)^{\mathsf T}\big)
\end{equation}
on $n = 2|G|$ qubits. Generalized bicycle codes~\cite{KovalevPryadko2013} are the case of a cyclic group, and bivariate bicycle codes~\cite{Bravyi2024} the case $G = \mathbb Z_\ell \times \mathbb Z_{\ell'}$ with $a$ and $b$ each a sum of three group elements.
\end{definition}

\begin{definition}[Coset-based two-block codes~\cite{AydinTamoBarg2026}]
\label{def:atb}
Let $H \le G$, and let $L$ and $R$ be the coset actions of Definition~\ref{def:coset-action}. For $a \in \mathbb F_2[G]$ and $b \in \mathbb F_2[N]$ the coset code of~\cite{AydinTamoBarg2026} is the
CSS code with
\begin{equation}
    H_X = \big(L(a) \,\big|\, R(b)\big), \qquad
    H_Z = \big(R(b)^{\mathsf T} \,\big|\, L(a)^{\mathsf T}\big)
\end{equation}
on $n = 2| G/H|$ qubits. 
\end{definition}
Note that for $H = \{e\}$, the above construction is the 2BGA code of $G$, and for $H \trianglelefteq G$, it is the 2BGA code of the quotient group $G/H$.

\begin{definition}[Lifted product codes]
\label{def:lp}
Let $\mathcal R_1, \mathcal R_2 \subseteq M_\ell(\mathbb F_2)$ be two subalgebras that commute elementwise, meaning $XY = YX$ for all $X \in \mathcal R_1$ and $Y \in \mathcal R_2$. Let $A$ be an $n_{A0} \times n_{A1}$ matrix with entries in $\mathcal R_1$ and $B$ an $n_{B0} \times n_{B1}$ matrix with entries in $\mathcal R_2$, viewed as
binary block matrices with $\ell \times \ell$ blocks. The lifted product code $\mathrm{LP}(A,B)$
of~\cite{PanteleevKalachev2021Degenerate,PanteleevKalachev2022LinearDistance} is the CSS code with
\begin{align}
    H_X &= \big(A \otimes I_{n_{B0}} \;\big|\; I_{n_{A0}} \otimes B\big),
    \label{eq:lp-hx}\\
    H_Z &= \big(I_{n_{A1}} \otimes B^{\mathsf T} \;\big|\;
              A^{\mathsf T} \otimes I_{n_{B1}}\big),
    \label{eq:lp-hz}
\end{align}
on $n = \ell\,(n_{A1}n_{B0} + n_{A0}n_{B1})$ qubits, where $\otimes$ is the Kronecker product taken with the $\ell \times \ell$ blocks as units, as in~\cite[Sec.~III-D]{PanteleevKalachev2022LinearDistance}, and $A^{\mathsf T}$, $B^{\mathsf T}$ are the transposes of the binary block matrices, which coincide with the conjugate transposes used
there.
\end{definition}

The case $\mathcal R_1 = \mathcal R_2 = \mathbb F_2$, with $\ell = 1$, is the hypergraph product~\cite{TillichZemor2014}, and $\mathcal R_1 = \mathcal R_2 = R$ for a commutative subalgebra $R$, with
the circulant algebra $\mathbb F_2[\mathbb Z_\ell]$ as the main example, is the setting of~\cite{PanteleevKalachev2021Degenerate,PanteleevKalachev2022LinearDistance}, where the non-abelian case is noted as possible and left open. That case is taken up in~\cite{PanteleevKalachev2022Good} with $\mathcal R_1 = L_G(\mathbb F_2[G])$ and $\mathcal R_2 = R_G(\mathbb F_2[G])$ for a non-abelian group $G$, and its $1 \times 1$ case is
Definition~\ref{def:2bga}.

\section{Construction of Coset Lifted Product Codes}
\label{sec:construction}

Throughout this section $G$ is a finite group, $H \le G$ is a subgroup
that need not be normal, $\Omega = G/H$, $m = |G|/|H|$, $N = N_G(H)$,
$L$ and $R$ are the coset actions of Section~\ref{sec:prelim-groups},
and $[n] = \{1, \dots, n\}$.

\subsection{Coset lifted product codes}

In this subsection, we define the codes, prove that they are CSS codes, and identify the earlier constructions among them.

\begin{definition}[Coset lifted product code]
\label{def:clp}
Let $A$ be an $n_{A0} \times n_{A1}$ matrix with entries $\alpha_{ij} \in \mathbb F_2[G]$ and $B$ an $n_{B0} \times n_{B1}$ matrix with entries $\beta_{kl} \in \mathbb F_2[N]$. Write $L(A)$ and $R(B)$ for the block matrices with $m \times m$ blocks $L(\alpha_{ij})$ and $R(\beta_{kl})$. The coset-lifted product code $\mathrm{CLP}(A,B;G,H)$
is the CSS code with
\begin{align}
    H_X &= \big(L(A) \otimes I_{n_{B0}} \;\big|\; I_{n_{A0}} \otimes R(B)\big),
    \label{eq:HX}\\
    H_Z &= \big(I_{n_{A1}} \otimes R(B)^{\mathsf T} \;\big|\;
              L(A)^{\mathsf T} \otimes I_{n_{B1}}\big)
        % = \big(I_{n_{A1}} \otimes R(B^{\dagger}) \;\big|\;
        %       L(A^{\dagger}) \otimes I_{n_{B1}}\big),
    \label{eq:HZ}
\end{align}
where $\otimes$ acts on $L(A)$ and $R(B)$ as block matrices. The code has $n = m\,(n_{A1}n_{B0} + n_{A0}n_{B1})$ qubits, $m\,n_{A0}n_{B0}$ checks of $X$ type, and $m\,n_{A1}n_{B1}$ checks of $Z$ type. 
\end{definition}

Before showing that Construction \ref{def:clp} satisfy CSS condition, we need the following lemma \cite{AydinTamoBarg2026}.

\begin{lemma}\label{lem:commute}
For every $g \in G$ and $n \in N$,
\begin{equation}
    L(g)\,R(n) = R(n)\,L(g).
\end{equation}
Consequently $L(\alpha)R(\beta) = R(\beta)L(\alpha)$ for all
$\alpha \in \mathbb F_2[G]$ and $\beta \in \mathbb F_2[N]$.
\end{lemma}

\begin{proof}
Let $xH \in \Omega$. By Definition~\ref{def:coset-action},
\begin{align}
    \big(L(g)\circ R(n)\big)(xH) &= L(g)\big(xnH\big) = gxnH, \\
    \big(R(n)\circ L(g)\big)(xH) &= R(n)\big(gxH\big) = gxnH .
\end{align}
The two permutations agree on every element of $\Omega$, so they are equal, and by Definition~\ref{def:permrep} so are their matrices. For the second claim write $\alpha = \sum_{g} a_g\, g$ and $\beta = \sum_{n} b_n\, n$ with $a_g, b_n \in \mathbb F_2$. Then
\begin{equation}
    L(\alpha)R(\beta)
    = \sum_{g}\sum_{n} a_g b_n\, L(g)R(n)
    = \sum_{g}\sum_{n} a_g b_n\, R(n)L(g)
    = R(\beta)L(\alpha).
\end{equation}
\end{proof}

\begin{theorem}[CSS orthogonality]
\label{thm:orthogonality}
For a finite group $G$, and subgroup $H \le G$, and every $A$ and $B$ as in Definition~\ref{def:clp}, $H_X H_Z^{\mathsf T} = 0$.
\end{theorem}

\begin{proof}
Index the block rows of $H_X$ by $(i,k) \in [n_{A0}] \times [n_{B0}]$ and the block rows of $H_Z$ by $(j,l) \in [n_{A1}] \times [n_{B1}]$. The block columns of both matrices fall into two groups, indexed by
$(j',k') \in [n_{A1}] \times [n_{B0}]$ and by $(i',l') \in [n_{A0}] \times [n_{B1}]$. By~\eqref{eq:HX} and
\eqref{eq:HZ}, the blocks are
\begin{align}
    (H_X)_{(i,k),(j',k')} &= \delta_{kk'}\, L(\alpha_{ij'}), &
    (H_X)_{(i,k),(i',l')} &= \delta_{ii'}\, R(\beta_{kl'}), \label{eq:HXblocks}\\
    (H_Z)_{(j,l),(j',k')} &= \delta_{jj'}\, R(\beta_{k'l})^{\mathsf T}, &
    (H_Z)_{(j,l),(i',l')} &= \delta_{ll'}\, L(\alpha_{i'j})^{\mathsf T}.
    \label{eq:HZblocks}
\end{align}
Since block transposition is the ordinary transposition of the binary matrix, $(H_Z^{\mathsf T})_{(j',k'),(j,l)} = \big((H_Z)_{(j,l),(j',k')}\big)^{\mathsf T}$ and likewise for the second column group. Multiplying blockwise and
summing over both column groups,
\begin{align}
    \big(H_X H_Z^{\mathsf T}\big)_{(i,k),(j,l)}
    &= \sum_{j',k'} \delta_{kk'}\, L(\alpha_{ij'})\,
       \big(\delta_{jj'}\, R(\beta_{k'l})^{\mathsf T}\big)^{\mathsf T}
     + \sum_{i',l'} \delta_{ii'}\, R(\beta_{kl'})\,
       \big(\delta_{ll'}\, L(\alpha_{i'j})^{\mathsf T}\big)^{\mathsf T} \nonumber\\
    &= L(\alpha_{ij})\, R(\beta_{kl}) + R(\beta_{kl})\, L(\alpha_{ij}) .
\end{align}
The two terms are equal by Lemma~\ref{lem:commute}, so their sum is zero over $\mathbb F_2$. Every block vanishes, hence
$H_X H_Z^{\mathsf T} = 0$.
\end{proof}

Definition~\ref{def:clp} and Theorem~\ref{thm:orthogonality} hold over any finite field $\mathbb F_q$ once the first block of $H_Z$ is negated, since the proof uses only Lemma~\ref{lem:commute}. The resulting codes
are qudit codes. All codes in this paper are binary.

\begin{proposition}[Special cases]
\label{prop:special-cases}
\begin{enumerate}
    \item If $H = \{e\}$, then $\Omega = G$, $N = G$, and $L$, $R$ are the
 left and right regular representations $L_G$, $R_G$. Hence $\mathrm{CLP}(A,B;G,\{e\})$ is the lifted product $\mathrm{LP}(A,B)$     of Definition~\ref{def:lp} with $\mathcal R_1 = L_G(\mathbb F_2[G])$ and $\mathcal R_2 = R_G(\mathbb F_2[G])$.
    \item If $n_{A0} = n_{A1} = n_{B0} = n_{B1} = 1$, then    $\mathrm{CLP}(A,B;G,H)$ is the coset-based two-block code of    Definition~\ref{def:atb}.
    \item If $H \trianglelefteq G$, let $\bar G = G/H$ and let $\bar A$, $\bar B$ be the images of $A$, $B$ under the map $\mathbb F_2[G] \to \mathbb F_2[\bar G]$ induced by $g \mapsto gH$. Then $\mathrm{CLP}(A,B;G,H) = \mathrm{CLP}(\bar A, \bar B; \bar G, \{e\})$,    a lifted product over $\bar G$ by statement 1.

        \item Let $K = \operatorname{Core}_G(H)$, and let $\bar A$, $\bar B$    be the images of $A$, $B$ under $\mathbb F_2[G] \to \mathbb F_2[G/K]$.    Then $\mathrm{CLP}(A,B;G,H) = \mathrm{CLP}(\bar A, \bar B; G/K, H/K)$, and $H/K$ has trivial core in $G/K$. Statement 3 is the case $K = H$.
\end{enumerate}

\end{proposition}

Proposition~\ref{prop:special-cases} restricts the pairs $(G,H)$ of interest as follows. By statement~4 of Proposition~\ref{prop:special-cases} every CLP code is equal to a CLP code with a core free subgroup; so $H$ may be taken core free. By statement 3 of Proposition~\ref{prop:special-cases}, a normal subgroup, $N = G$, yields only lifted products over the quotient group. At the other extreme, a self normalizing subgroup, $N = H$, makes every $R(n)$ the identity by Proposition~\ref{prop:LR}, so that $R(\beta)$ is $0$ or $I_m$ for every $\beta \in \mathbb F_2[N]$ and the right protograph contributes only the parities of its entry weights. The pairs that can produce new codes are therefore those with $H$ core free and $H < N < G$.

Next, we explain our construction using an example.
% Next we illustrate the construction on a small example.

\begin{example}
\label{ex:a4}
Let $G = A_4$ be the alternating group, a subgroup of $S_4$, consisting
of all the even permutations of $\{1,2,3,4\}$, and let
$H = \langle (12)(34) \rangle = \{e, (12)(34)\}$. Notice that $H$ is not a normal subgroup of $A_4$.

We have $|G| = 12$, $|H| = 2$, $m = 6$,
$\operatorname{Core}_G(H) = \{e\}$, and
\[
    N = N_G(H) = \{e,\ (12)(34),\ (13)(24),\ (14)(23)\} \cong
    \mathbb Z_2 \times \mathbb Z_2 .
\]

The quotient of $G$ by $H$ gives total $|G/H|=6$ cosets, 
\[
    \omega_1 = eH,\quad \omega_2 = (234)H,\quad \omega_3 = (243)H,\quad
    \omega_4 = (123)H,\quad \omega_5 = (124)H,\quad \omega_6 = (13)(24)H.
\]
So $\Omega = \{\omega_1, \dots, \omega_6\}$, and each $\omega_i$ contains two group elements.
By Proposition~\ref{prop:LR}, $L$ takes $|G|/|\operatorname{Core}_G(H)| = 12$
distinct values and $R$ takes $|N|/|H| = 2$ distinct values.

Let protograph matrices are
\[
    A = \begin{pmatrix} e & e + (234) \\ e + (124) & e + (134) \end{pmatrix}
    \in M_{2\times 2}(\mathbb F_2[G]),
    \qquad
    B = \begin{pmatrix} e & e \\ e & (13)(24) \end{pmatrix}
    \in M_{2\times 2}(\mathbb F_2[N]),
\]
so that, with $I = I_6$,
\[
    L(A) = \begin{pmatrix} I & I + L\big((234)\big) \\
                           I + L\big((124)\big) & I + L\big((134)\big) \end{pmatrix},
    \qquad
    R(B) = \begin{pmatrix} I & I \\ I & R\big((13)(24)\big) \end{pmatrix}.
\]

The blocks of $L(A)$ and $R(B)$ are built from~\eqref{eq:matconv}. For
instance, $(234)$ acts on the cosets by
\[
    \omega_1 = eH \mapsto (234)H = \omega_2, \quad
    \omega_2 \mapsto (234)^2H = \omega_3, \quad
    \omega_3 \mapsto \omega_1, \qquad
    \omega_4 \mapsto \omega_6, \quad
    \omega_6 \mapsto \omega_5, \quad
    \omega_5 \mapsto \omega_4,
\]
that is, $L((234)) = (1\,2\,3)(4\,6\,5)$ as a permutation of the indices.
By~\eqref{eq:matconv}, column $j$ of its matrix has its single $1$ in
the row of the image of $\omega_j$, and adding $L(e) = I_6$ gives the
block
\[
    L\big(e + (234)\big) = I_6 + L\big((234)\big) =
    \begin{pmatrix}
        1&0&1&0&0&0\\
        1&1&0&0&0&0\\
        0&1&1&0&0&0\\
        0&0&0&1&1&0\\
        0&0&0&0&1&1\\
        0&0&0&1&0&1
    \end{pmatrix}.
\]
In the same way, we can obtain all the blocks of $L(A)$ and $ R(B) $.

Now, by Definition~\ref{def:clp} with
$n_{A0} = n_{A1} = n_{B0} = n_{B1} = 2$,
\[
    H_X = \big(L(A) \otimes I_2 \;\big|\; I_2 \otimes R(B)\big),
    \qquad
    H_Z = \big(I_2 \otimes R(B)^{\mathsf T} \;\big|\; L(A)^{\mathsf T} \otimes I_2\big),
\]
each of size $24 \times 48$, on $n = 6\,(2\cdot 2 + 2\cdot 2) = 48$
qubits. Then $H_X H_Z^{\mathsf T} = 0$ by
Theorem~\ref{thm:orthogonality}. Gaussian elimination gives
$\operatorname{rank} H_X = \operatorname{rank} H_Z = 21$, so $k = 6$, and
we calculate distance using the integer linear program of Section \ref{sec:cert} to be $d_X = d_Z = 4$. The code is a $[[48,6,4]]$.
\end{example}

Different protographs can define the same code up to a relabelling of
qubits and checks. Two CSS codes are \emph{equivalent} if their check
matrices satisfy $H_X' = P_X H_X Q$ and $H_Z' = P_Z H_Z Q$ for
permutation matrices $Q$ on the qubits and $P_X$, $P_Z$ on the checks.
Equivalent codes have the same parameters and the same check weights.
The following lemma extends the equivalences of
\cite[Lemmas~IV.1 and IV.2]{AydinTamoBarg2026}, which are the case of
$1 \times 1$ protographs, to CLP codes.

\begin{lemma}[Equivalent protographs]
\label{lem:equiv}
Let $A$, $B$ be as in Definition~\ref{def:clp}.
\begin{enumerate}
\item Let $g_i, g'_j \in G$ for $i \in [n_{A0}]$, $j \in [n_{A1}]$, and
$h_k, h'_l \in N$ for $k \in [n_{B0}]$, $l \in [n_{B1}]$, and let $A'$,
$B'$ have entries
\[
    \alpha'_{ij} = g_i\,\alpha_{ij}\,g'_j, \qquad
    \beta'_{kl} = h'_l\,\beta_{kl}\,h_k .
\]
Then $\mathrm{CLP}(A,B;G,H)$ and $\mathrm{CLP}(A',B';G,H)$ are equivalent.
\item Let $\varphi$ be an automorphism of $G$, and let $\varphi(A)$,
$\varphi(B)$, $\varphi(H)$ denote its application to every group
element. Then $\mathrm{CLP}(A,B;G,H)$ and
$\mathrm{CLP}(\varphi(A),\varphi(B);G,\varphi(H))$ are equivalent.
\end{enumerate}
\end{lemma}
% By the definition of equivalence, it suffices to show that the matrices $H_X$ and $H_Z$ obtained by these replacements are  $P_X H_X Q$ and $P_Z H_Z Q$ for permutation matrices $P_X$, $P_Z$ on the checks and $Q$ on the qubits.
\begin{proof}
Let $H_X$, $H_Z$ be the check matrices of $\mathrm{CLP}(A,B;G,H)$ and
$H_X'$, $H_Z'$ those of the replaced code. Throughout we use the block
formulas
\[
    (H_X)_{(i,k),(j,k')} = \delta_{kk'} L(\alpha_{ij}), \quad
    (H_X)_{(i,k),(i',l')} = \delta_{ii'} R(\beta_{kl'}), \quad \]
  \[  (H_Z)_{(j,l),(j',k')} = \delta_{jj'} R(\beta_{k'l})^{\mathsf T}, \quad
    (H_Z)_{(j,l),(i',l')} = \delta_{ll'} L(\alpha_{i'j})^{\mathsf T}
\]
from the proof of Theorem~\ref{thm:orthogonality}, and we write
$\operatorname{diag}_{(i,k)}(M_{ik})$ for the block diagonal matrix with
block $M_{ik}$ at the position indexed by $(i,k)$.

1) By Proposition~\ref{prop:LR},
\[
    L(\alpha'_{ij}) = L(g_i)\,L(\alpha_{ij})\,L(g'_j), \qquad
    R(\beta'_{kl}) = R(h_k)\,R(\beta_{kl})\,R(h'_l),
\]
the second because $R$ reverses products. Define the permutation
matrices
\[
    Q = \operatorname{diag}\Big(
        \operatorname{diag}_{(j,k')}\big(L(g'_j)R(h_{k'})^{-1}\big),\
        \operatorname{diag}_{(i',l')}\big(L(g_{i'})^{-1}R(h'_{l'})\big)
        \Big)
\]
on the two qubit groups,
$M_X = \operatorname{diag}_{(i,k)}\big(L(g_i)R(h_k)\big)$ on the $X$
checks, and
$M_Z = \operatorname{diag}_{(j,l)}\big(L(g'_j)^{-1}R(h'_l)^{-1}\big)$ on
the $Z$ checks. Each block is a product of permutation matrices, so
all three are permutation matrices. Multiplying out block by block and
moving every $L$ past every $R$ with Lemma~\ref{lem:commute},
\[
    (M_X H_X Q)_{(i,k),(j,k')}
    = \delta_{kk'}\, L(g_i)R(h_k)\,L(\alpha_{ij})\,L(g'_j)R(h_k)^{-1}
    = \delta_{kk'}\, L(g_i)L(\alpha_{ij})L(g'_j)
    = \delta_{kk'}\, L(\alpha'_{ij}),
\]
\[
    (M_X H_X Q)_{(i,k),(i',l')}
    = \delta_{ii'}\, L(g_i)R(h_k)\,R(\beta_{kl'})\,L(g_i)^{-1}R(h'_{l'})
    = \delta_{ii'}\, R(h_k)R(\beta_{kl'})R(h'_{l'})
    = \delta_{ii'}\, R(\beta'_{kl'}),
\]
so $M_X H_X Q = H_X'$. In the same way, using
$M^{\mathsf T} = M^{-1}$ for permutation matrices,
\[
    (M_Z H_Z Q)_{(j,l),(j',k')}
    = \delta_{jj'}\, R(h'_l)^{-1}R(\beta_{k'l})^{\mathsf T}R(h_{k'})^{-1}
    = \delta_{jj'}\, \big(R(h_{k'})R(\beta_{k'l})R(h'_l)\big)^{\mathsf T}
    = \delta_{jj'}\, R(\beta'_{k'l})^{\mathsf T},
\]
\[
    (M_Z H_Z Q)_{(j,l),(i',l')}
    = \delta_{ll'}\, L(g'_j)^{-1}L(\alpha_{i'j})^{\mathsf T}L(g_{i'})^{-1}
    = \delta_{ll'}\, \big(L(g_{i'})L(\alpha_{i'j})L(g'_j)\big)^{\mathsf T}
    = \delta_{ll'}\, L(\alpha'_{i'j})^{\mathsf T},
\]
so $M_Z H_Z Q = H_Z'$. Hence $H_X' = M_X H_X Q$ and $H_Z' = M_Z H_Z Q$,
with the same qubit permutation $Q$ on both sides.

2) Let $H' = \varphi(H)$, so that $N_G(H') = \varphi(N)$, and define
\[
    \sigma \colon G/H \to G/H', \qquad xH \mapsto \varphi(x)H' .
\]
This map is well defined and bijective, since
$\varphi(xh)H' = \varphi(x)\varphi(h)H' = \varphi(x)H'$ for $h \in H$
and $\varphi$ is a bijection of $G$. For $y \in G$ and $n \in N$,
\[
    \sigma\big(L(y)(xH)\big) = \varphi(y)\varphi(x)H' = L'(\varphi(y))\big(\sigma(xH)\big),
    \qquad
    \sigma\big(R(n)(xH)\big) = \varphi(x)\varphi(n)H' = R'(\varphi(n))\big(\sigma(xH)\big),
\]
where $L'$, $R'$ are the coset actions for $H'$. With $P_\sigma$ the
permutation matrix of $\sigma$ this reads
\[
    L'(\varphi(y)) = P_\sigma\, L(y)\, P_\sigma^{\mathsf T}, \qquad
    R'(\varphi(n)) = P_\sigma\, R(n)\, P_\sigma^{\mathsf T},
\]
and by linearity the same holds for $L(\alpha)$ and $R(\beta)$. Taking
$P_X$, $P_Z$, and $Q$ to be block diagonal with every block equal to
$P_\sigma$ therefore gives $P_X H_X Q = H_X'$ and $P_Z H_Z Q = H_Z'$.
\end{proof}

For an inner automorphism $\varphi(x) = g^{-1}xg$, part 2 replaces $A$
by $g^{-1}Ag$, and part 1 with $g_i = g$, $g'_j = g^{-1}$ undoes this,
so $\mathrm{CLP}(A,B;G,H)$ is also equivalent to
$\mathrm{CLP}(A, g^{-1}Bg; G, g^{-1}Hg)$, which is
\cite[Lemma~IV.2]{AydinTamoBarg2026} for $1 \times 1$ protographs.

\subsection{Dimension}
\label{sec:dimension}

No closed formula for the dimension of a lifted product code is known
in general \cite[Sec.~III-D]{PanteleevKalachev2022LinearDistance}. For
abelian groups of odd order, Panteleev and Kalachev observe that the
group algebra splits into a product of finite fields, that the lifted
product splits accordingly into a direct sum of hypergraph product
codes over those fields, and that the dimension is therefore the sum
of the dimensions of these hypergraph products, each of which has a
closed formula \cite[App.~B]{PanteleevKalachev2022LinearDistance}. 

For
balanced products over a group of odd order, which include lifted
products with regular actions of non-abelian groups, the dimension
follows from the K\"unneth formula~\cite[Lemma~19]{BreuckmannEberhardt2021Balanced} in terms of the
homology of the two factors, and for two-block codes over such groups
it is shown in~\cite{WangLinPryadko2023} to be even. In this subsection
we treat CLP codes for every group of odd order, abelian or not, and
every subgroup $H$. The fields of the abelian case are replaced by
irreducible representations, the subgroup enters through the dimension
of its fixed space in each representation, and the dimension is
expressed through matrices of size $\dim\rho$ and $\dim\rho^H$ rather
than through the homology of the factors, which for a coset action is
not free.

\subsubsection{Preliminaries from representation theory}

Throughout this subsection $|G|$ is odd and $\mathbb F$ is an algebraic
closure of $\mathbb F_2$. A \emph{representation} of $G$ over
$\mathbb F$ is a homomorphism $\rho \colon G \to \operatorname{GL}_{d}(\mathbb F)$,
extended linearly to $\mathbb F[G]$ by
$\rho(\sum_g a_g g) = \sum_g a_g \rho(g)$. We write $d_\rho = d$ for
its dimension. Two representations are equivalent if they differ by a
change of basis, and $\rho$ is \emph{irreducible} if the only
subspaces of $\mathbb F^{d}$ mapped to themselves by every $\rho(g)$
are $0$ and $\mathbb F^{d}$. From now on $\rho$ ranges over the
irreducible representations of $G$ over $\mathbb F$, one from each
equivalence class. For a subgroup $H \le G$ let
\[
    \rho^H = \{\, v \in \mathbb F^{d_\rho} \mid \rho(h)v = v \text{ for all } h \in H \,\}
\]
be the subspace of vectors fixed by $H$. Since $|G|$ is odd, the characteristic of field does not divide $|G|$, by  Maschke's theorem the following facts hold, see~\cite{Serre1977} for more details.

\begin{enumerate}
\item[(F1)] \emph{Decomposition of the coset representation.} There
exist an invertible matrix $S \in \operatorname{GL}_m(\mathbb F)$ and
uniquely determined nonnegative integers $\mu_\rho$ such that, for
every $g \in G$,
\begin{equation}\label{eq:decomp}
    S^{-1} L(g)\, S = \bigoplus_\rho \big(\rho(g) \otimes I_{\mu_\rho}\big),
\end{equation}
where $\bigoplus$ denotes a block diagonal matrix with one block per
$\rho$. In words, the matrices $L(g)$ can be simultaneously brought
to a block diagonal form in which $\rho(g)$ appears $\mu_\rho$ times.
By linearity~\eqref{eq:decomp} holds with $\alpha \in \mathbb F[G]$ in
place of $g$, and $\sum_\rho d_\rho\,\mu_\rho = m$.
\item[(F2)] \emph{Multiplicities.} The number of times $\rho$ appears
is
\begin{equation}\label{eq:mult}
    \mu_\rho = \dim \rho^H .
\end{equation}
This is Frobenius reciprocity applied to the coset representation,
which is induced from the trivial representation of $H$.
\item[(F3)] \emph{Schur's lemma.} If $T \in M_m(\mathbb F)$ commutes
with every $L(g)$, then
\begin{equation}\label{eq:schur}
    S^{-1} T\, S = \bigoplus_\rho \big(I_{d_\rho} \otimes T_\rho\big)
\end{equation}
% for unique matrices $T_\rho \in M_{\mu_\rho}(\mathbb F)$. In words, a
% matrix commuting with the whole representation cannot act inside a
% copy of $\rho$, it can only mix the $\mu_\rho$ copies.

In block form, with two irreducible representations $\rho_1$, $\rho_2$
for illustration, \eqref{eq:decomp} and~\eqref{eq:schur} read
\[
    S^{-1} L(\alpha)\, S =
    \begin{pmatrix}
        \rho_1(\alpha) \otimes I_{\mu_1} & 0 \\
        0 & \rho_2(\alpha) \otimes I_{\mu_2}
    \end{pmatrix},
    \qquad
    S^{-1} R(\beta)\, S =
    \begin{pmatrix}
        I_{d_1} \otimes \tau_1(\beta) & 0 \\
        0 & I_{d_2} \otimes \tau_2(\beta)
    \end{pmatrix}.
\]
% The left factor acts inside each copy of $\rho$ and never mixes
% copies. The right factor mixes the copies and never acts inside one.
% This is the only reason the two factors can be separated.

\item[(F4)] \emph{Rank under field extension.} A matrix with entries
in $\mathbb F_2$ has the same rank over $\mathbb F_2$ and over $\mathbb F$.
\end{enumerate}

We also use the dimension of a hypergraph product. For matrices $A'$
of size $r_A \times c_A$ and $B'$ of size $r_B \times c_B$ over any
field, $\mathrm{HGP}(A',B')$ is Definition~\ref{def:lp} with
$\ell = 1$. It has length $c_A r_B + r_A c_B$ and dimension
\begin{equation}\label{eq:hgpdim}
    k_{\mathrm{HGP}}(A',B') = \kappa(A')\,\kappa(B'^{\mathsf T})
                            + \kappa(A'^{\mathsf T})\,\kappa(B'),
\end{equation}
where $\kappa(M)$ is the dimension of the kernel of $M$
\cite[Eq.~(10)]{PanteleevKalachev2022LinearDistance}.

\subsubsection{The matrices $A_\rho$ and $B_\rho$}

For a matrix $A$ over $\mathbb F_2[G]$ with entries $\alpha_{ij}$ set
\begin{equation}\label{eq:Arho}
    A_\rho = \big(\rho(\alpha_{ij})\big)_{i,j}
    \in M_{n_{A0} d_\rho \times n_{A1} d_\rho}(\mathbb F),
\end{equation}
the matrix obtained from $A$ by replacing every entry by its image
under $\rho$. For $\beta \in \mathbb F_2[N]$ the matrix $R(\beta)$
commutes with every $L(g)$ by Lemma~\ref{lem:commute}; so, (F3) applies
to $T = R(\beta)$. Write $\tau_\rho(\beta)$ for the resulting block
$T_\rho$, so that
\begin{equation}\label{eq:tau}
    S^{-1} R(\beta)\, S = \bigoplus_\rho \big(I_{d_\rho} \otimes \tau_\rho(\beta)\big),
\end{equation}
and for a matrix $B$ over $\mathbb F_2[N]$ with entries $\beta_{kl}$ set
\begin{equation}\label{eq:Brho}
    B_\rho = \big(\tau_\rho(\beta_{kl})\big)_{k,l}
    \in M_{n_{B0}\mu_\rho \times n_{B1}\mu_\rho}(\mathbb F).
\end{equation}

\begin{theorem}[Dimension for odd order]
\label{thm:odd}
Let $|G|$ be odd and let $A$, $B$ be as in Definition~\ref{def:clp}.
Then
\begin{equation}\label{eq:kodd}
    k \;=\; \sum_{\rho} \Big( \kappa(A_\rho)\,\kappa(B_\rho^{\mathsf T})
                            + \kappa(A_\rho^{\mathsf T})\,\kappa(B_\rho) \Big),
\end{equation}
the sum running over the irreducible representations with
$\mu_\rho > 0$. Equivalently, $\mathrm{CLP}(A,B;G,H)$ has the same
dimension as the direct sum of the hypergraph product codes
$\mathrm{HGP}(A_\rho, B_\rho)$ over $\mathbb F$.
\end{theorem}

\begin{proof}
By Definition~\ref{def:css} and (F4), $k = n - \operatorname{rank} H_X
- \operatorname{rank} H_Z$ with both ranks computed over $\mathbb F$.
We change basis by $S$ in every block and read off the ranks.

\emph{Step 1, change of basis.} Let $S_q$ be the block diagonal matrix
with one copy of $S$ for every block of qubits and $S_c$ the block
diagonal matrix with one copy of $S$ for every block of $X$ checks.
By~\eqref{eq:decomp} and~\eqref{eq:tau}, every block $L(\alpha_{ij})$
of $S_c^{-1} H_X S_q$ becomes $\bigoplus_\rho (\rho(\alpha_{ij}) \otimes I_{\mu_\rho})$
and every block $R(\beta_{kl})$ becomes
$\bigoplus_\rho (I_{d_\rho} \otimes \tau_\rho(\beta_{kl}))$.

\emph{Step 2, the $X$ checks.} Group together, for each $\rho$, the
$\rho$ parts of all blocks. This is a permutation of rows and columns,
after which $S_c^{-1} H_X S_q$ is block diagonal with one block per
$\rho$, namely
\begin{equation}\label{eq:HXrho}
    H_X^{\rho} = \Big( A_\rho \otimes I_{\mu_\rho n_{B0}} \;\Big|\;
                       I_{n_{A0} d_\rho} \otimes B_\rho \Big),
\end{equation}
which is the $X$ check matrix of $\mathrm{HGP}(A_\rho, B_\rho)$. Since
changes of basis and permutations preserve rank, we obtain
$\operatorname{rank} H_X = \sum_\rho \operatorname{rank} H_X^\rho$.

\emph{Step 3, the $Z$ checks.} Transposing~\eqref{eq:HZ} gives
\[
    H_Z^{\mathsf T} = \begin{pmatrix} I_{n_{A1}} \otimes R(B) \\ L(A) \otimes I_{n_{B1}} \end{pmatrix},
\]
which contains $L(A)$ and $R(B)$ without transposes. Steps 1 and 2
applied to this matrix give, for each $\rho$, the block
$\big(I_{n_{A1} d_\rho} \otimes B_\rho \,;\, A_\rho \otimes I_{\mu_\rho n_{B1}}\big)$,
whose transpose
\begin{equation}\label{eq:HZrho}
    H_Z^{\rho} = \Big( I_{n_{A1} d_\rho} \otimes B_\rho^{\mathsf T} \;\Big|\;
                       A_\rho^{\mathsf T} \otimes I_{\mu_\rho n_{B1}} \Big)
\end{equation}
is the $Z$ check matrix of $\mathrm{HGP}(A_\rho, B_\rho)$. Hence
$\operatorname{rank} H_Z = \operatorname{rank} H_Z^{\mathsf T}
= \sum_\rho \operatorname{rank} H_Z^\rho$.

\emph{Step 4, the count.} The $\rho$ block has
$d_\rho \mu_\rho (n_{A1}n_{B0} + n_{A0}n_{B1})$ columns, the length of
$\mathrm{HGP}(A_\rho, B_\rho)$, and these sum to $n$ because
$\sum_\rho d_\rho \mu_\rho = m$. Therefore
$k = \sum_\rho k_{\mathrm{HGP}}(A_\rho, B_\rho)$, and~\eqref{eq:kodd}
follows from~\eqref{eq:hgpdim}.
\end{proof}

\begin{remark}[Computing $\tau_\rho$]
\label{rem:tau}
The matrices $\tau_\rho(\beta)$ can be obtained without computing
$S$. Let $W_\rho \subseteq \mathbb F^{1 \times d_\rho}$ be the space of
row vectors $f$ with $f\rho(h) = f$ for all $h \in H$. Then
$\dim W_\rho = \mu_\rho$, every $n \in N$ maps $W_\rho$ to itself by
$f \mapsto f\rho(n)$ because $n$ normalizes $H$, and $\tau_\rho(n)$ is
the matrix of this map in any basis of $W_\rho$, extended linearly to
$\mathbb F_2[N]$. In particular $\tau_\rho(nn') = \tau_\rho(n')\tau_\rho(n)$,
matching the order reversing property of $R$ in
Proposition~\ref{prop:LR}. This is the explicit form of the Frobenius
reciprocity isomorphism behind (F2), and we omit the routine
verification.
\end{remark}

\begin{example}
\label{ex:odd}
Let $G$ be the Heisenberg group of order $27$, realized as the set of
triples $(a,b,c) \in \mathbb F_3^3$ with product
$(a,b,c)(a',b',c') = (a+a',\ b+b',\ c+c'+ab')$, and let $x = (1,0,0)$,
$y = (0,1,0)$, $z = (0,0,1)$. The centre is $\langle z \rangle$ and
$G/\langle z\rangle \cong \mathbb Z_3^2$ through $(a,b,c) \mapsto (a,b)$.
Let $H = \langle x \rangle$. Then $m = 9$, $\operatorname{Core}_G(H)$ is
trivial, and $N = \{(a,0,c)\} = \langle x, z\rangle$, so
$N/H \cong \mathbb Z_3$ is generated by $zH$. Over
$\mathbb F_4 = \{0,1,\omega,\omega^2\}$ the group has nine one
dimensional modules $\chi_{s,t}$, $s,t \in \mathbb F_3$, given by
$\chi_{s,t}(a,b,c) = \omega^{sa+tb}$, and two three dimensional modules
$\rho_1$, $\rho_2$, on which $z$ acts as the scalar $\omega$ and
$\omega^2$ respectively. By~\eqref{eq:mult}, $\mu_{\chi_{s,t}} = 1$
if $s = 0$ and $0$ otherwise, since $\chi_{s,t}(x) = \omega^s$, and
$\mu_{\rho_i} = 1$, since the restriction of $\rho_i$ to $H$ is the
regular module of $\mathbb Z_3$. As a check,
$3 \cdot 1 + 2 \cdot 3 = 9 = m$.

Take $A = (e+z \ \ e+y)$ and $B = (e+z \ \ z^2)^{\mathsf T}$, so
$n = 9\,(2\cdot 2 + 1\cdot 1) = 45$. Direct elimination gives
$\operatorname{rank} H_X = 17$, $\operatorname{rank} H_Z = 18$, and
$k = 10$. Theorem~\ref{thm:odd} gives the same value term by term.
For $\chi_{0,t}$ one has $\chi_{0,t}(z) = 1$, so
$A_{\chi_{0,t}} = (0 \ \ 1+\omega^t)$ and, since $\tau_\rho(z) = 1$
on the one dimensional space $(\rho^*)^H$,
$B_{\chi_{0,t}} = (0 \ \ 1)^{\mathsf T}$. Formula~\eqref{eq:hgpdim} gives
$k_{\chi_{0,0}} = 2\cdot 1 + 1\cdot 0 = 2$ and
$k_{\chi_{0,1}} = k_{\chi_{0,2}} = 1\cdot 1 + 0\cdot 0 = 1$. For
$\rho_i$ one has $A_{\rho_i} = \big((1+\omega^i)I_3 \ \ I_3 + \rho_i(y)\big)$
of rank $3$ and $B_{\rho_i} = (1+\omega^i \ \ \omega^{2i})^{\mathsf T}$ of
rank $1$, so $k_{\rho_i} = 3\cdot 1 + 0\cdot 0 = 3$. The sum is
$2 + 1 + 1 + 3 + 3 = 10$. The five blocks of~\eqref{eq:HXrho}
and~\eqref{eq:HZrho} have ranks $(1,2,2,6,6)$ and $(2,2,2,6,6)$, whose
sums are $17$ and $18$.
\end{example}

\section{Search for codes}
\label{sec:search}

The search produced certified codes at check weights $5$, $6$, $7$,
and $9$, and candidates with distance upper bounds at weights $7$,
$9$, and $10$. Following~\cite{qian2026multi}, codes are compared at
fixed check weight $w$ by $Q = kd^2/n$, with certified distances and
upper bounds kept apart. At $w = 6$ the $[[360,20,20]]$ code attains
$Q = 22.22$, against $15.25$ for the $[[340,16,18]]$ generalized
bicycle code cited as the leading exact code in~\cite{qian2026multi}
and $13.71$ for the $[[224,12,16]]$ code of~\cite{AydinTamoBarg2026}.
At $w = 7$ the $[[360,24,20]]$ code attains $Q = 26.67$, against
$18.00$ for the certified $[[288,16,18]]$ code and $24.00$ for the
bound based $[[384,16,\le 24]]$ code of~\cite{qian2026multi}. At
$w = 9$ the $[[252,16,18]]$ code attains $Q = 20.57$, against $20.25$
for the $[[288,18,18]]$ code of~\cite{qian2026multi}. At $w = 10$ the
candidate $[[360,20,\le 27]]$ reaches $Q \le 40.50$, above the
certified $38.77$ of~\cite{qian2026multi} but below their bound based
endpoint. Our candidates at $w = 9$ reach $Q \le 32.03$ and stay
below the bound based endpoint $36.75$ of~\cite{qian2026multi}. No
code at $w = 8$ cleared the certified figure $30.56$. All codes are
listed in Tables~\ref{tab:cert} and~\ref{tab:cand}, and the
protographs are given in Appendix~\ref{app:data}.

\subsection{Conventions and search space}
\label{sec:space}

Groups are named by their SmallGroups identifier~\cite{SmallGroups,GAP},
elements are words in the polycyclic generating sequence
$f_1, \dots, f_r$ returned by \texttt{SmallGroup}$(\ell,i)$ in
GAP~4.15.1, and $|a|$ is the number of words in $a$. Every code was
rebuilt from the printed data by Definition~\ref{def:clp} and
compared with the stored check matrices.

$[[n,k,d]]$ means that $d$ is
certified by the integer program of Section~\ref{sec:cert}, and
$[[n,k,\le d]]$ that a $10^6$ round search with
QDistRnd~\cite{QDistRnd} on each side found weight $d$ and nothing
lighter. Throughout, "$p\times q$ over $r\times s$" means that $A$ has size
$p\times q$ and $B$ has size $r\times s$.

Two layouts occur. In the $1 \times 1$ over $2 \times 2$ layout
$A = [a]$ and $B$ is $2 \times 2$. In the $2 \times 2$ over $1 \times 1$
layout $A$ is $2 \times 2$ and $B = [b]$. Both give $n = 4m$ with
$m = |G|/|H|$. The search space was restricted as follows.
\begin{enumerate}
\item Pairs $(G,H)$ with $[H\,{:}\,\Core_G(H)] \le 3$. By
Proposition~\ref{prop:special-cases} item 4, such a code equals one over
$G/\Core_G(H)$ with a core free subgroup of order $2$ or $3$, and the
SmallGroups identifiers of the printed data are kept.
\item $[G\,{:}\,N] \in \{2,3,4\}$ and $N \ne H$. Normal subgroups give
lifted products by Proposition~\ref{prop:special-cases} item 3, and
$N = H$ makes every $R(n)$ the identity by Proposition~\ref{prop:LR}.
\item $30 \le m \le 100$, so that $120 \le n \le 400$, over the range of the
tables in~\cite{AydinTamoBarg2026,qian2026multi}. This gave $14,843$
pairs.
\item The support of $a$ contains the identity and an element of
$G \setminus N$, and three entries of $B$ contain the identity. The
first and third conditions are normal forms by Lemma~\ref{lem:equiv}.
The second excludes $a \in \Ftwo[N]$, for which the code is a direct
sum over the orbits of $N$ on $G/H$.
\item For a target weight $w$, $|a| \in [3, w-3]$, the entries of $B$
have weights $1$ to $3$, and the largest row or column sum of entry
weights of $B$ is at most $w- |a|$.
\end{enumerate}
% \textcolor{red}{An earlier run on the pairs behind the codes of~\cite{AydinTamoBarg2026}
% tried the layouts $1\times1$ over $1\times1$, $2\times2$ over
% $1\times1$, $1\times1$ over $2\times2$, and $2\times2$ over
% $2\times2$, with $|a| \le 3$ and entries of weight at most $2$. The
% first recovered the $[[56,6,8]]$ and $[[36,4,6]]$ codes
% of~\cite{AydinTamoBarg2026}, which validated the
% pipeline. The $2\times2$ over $1\times1$ layout produced four
% weight $5$ codes, and the $[[216,8,18]]$ code of Table~\ref{tab:cert}.
% Layouts with a $1\times2$ or $2\times1$ factor never exceeded
% $d_X = 4$ once both sides were screened. All codes above the published
% figures come from the $1\times1$ over $2\times2$ layout. SGS: Please rewrite these statements with clarity.}

 Before the main search, an exploratory run on the pairs $(G,H)$ behind the codes  of~\cite{AydinTamoBarg2026} tried $1\times1$ over $1\times1$, $2\times2$ over $1\times1$, $1\times1$ over $2\times2$, and $2\times2$ over $2\times2$, with entries of $A$ of weight at most $3$ and entries of $B$ of weight at most $2$. The first is the construction of~\cite{AydinTamoBarg2026} itself, and the run recovered their $[[56,6,8]]$ and $[[36,4,6]]$ codes, which validated the pipeline. The second, a $2\times2$ matrix $A$ over $\Ftwo[G]$ and a single element $b \in \Ftwo[N]$, produced the four codes marked $\ast$ in Table~\ref{tab:cert}, three of weight $5$ and the $[[216,8,18]]$ code. The fourth produced nothing. 
 
 Four further shapes were then run across all pairs and closed.
$1\times2$ over $2\times1$ and $1\times1$ over $2\times3$ never
exceeded $d_X = 4$ once both sides were screened. $3\times3$ over
$1\times1$ places the only heavy element $b$ on the $N$ side, and
$R(b)$ was singular for most draws, with $d \le 9$ throughout.
$1\times1$ over $3\times3$ with single element entries reached
$k = 36$ at $n \le 288$ but never $d > 10$. Every code above the
published figures comes from $1\times1$ over $2\times2$, a single
element $a \in \Ftwo[G]$ and a $2\times2$ matrix $B$ over $\Ftwo[N]$,
and the main search was restricted to it.

\subsection{Pipeline and certification}
\label{sec:cert}

For each pair and target weight $W \in \{6, \dots, 10\}$, $5\,000$
random $(a,B)$ were drawn under the conditions above. For each draw,
$H_XH_Z^{\mathsf T} = 0$ was verified numerically, $k$ was computed
from the ranks, draws with $k < 8$ or with a cancellation in a product
$L(a)R(\beta_{kl})$ were discarded, and $d_X$, $d_Z$ were estimated
with QDistRnd at $500$, $20\,000$, and $200\,000$ rounds. A draw
survived only if both estimates stayed at or above the distance needed
for $Q$ to exceed the certified figure at its realized weight. Survivors
were rescreened at $10^6$ rounds on both sides. Of the seven
candidates with a bound of $25$ or more after $2 \times 10^5$ rounds,
six lost between one and five at $10^6$ rounds, so bounds above $25$
should be read with care.

Codes with $d \le 20$ were certified exactly. Let
$\bar L^Z_1, \dots, \bar L^Z_k$ be a basis of logical $Z$ operators
modulo stabilizers. Every nontrivial logical $X$ operator anticommutes
with some $\bar L^Z_j$, so
\begin{equation}\label{eq:sector}
    d_X = \min_{j \in [k]}\ \min\big\{\, \operatorname{wt}(x) \ \big|\
    x \in \Ftwo^n,\ H_Z x = 0,\ \langle x, \bar L^Z_j \rangle = 1 \,\big\},
\end{equation}
and symmetrically for $d_Z$. Each inner minimum is a binary integer
program~\cite{LAR2011}, hard in general~\cite{KapshikarKundu2023},
solved in Gurobi~13.0.3~\cite{Gurobi} with a cutoff at
$d_{\mathrm{ub}} - \tfrac12$. A sector that terminates at the cutoff
has no solution of weight below $d_{\mathrm{ub}}$, and when every
sector on both sides terminates this way, $d \ge d_{\mathrm{ub}}$. One
sector solved without the cutoff returns a weight $d_{\mathrm{ub}}$
logical operator, which pins $d$ exactly. A separate NumPy script
rechecks orthogonality, dimension, check weights, and the witnesses
from the matrix files alone for sanity check. 
% \chk{hardware and wall time for the
% largest certified code}

\subsection{Certified CLP codes and candidates search results}
\label{sec:codes}

Table~\ref{tab:cert} lists every code with certified distance and
Table~\ref{tab:cand} every candidate whose bound held at $10^6$ rounds
on both sides and whose $Q$ at the bound exceeds the certified figure
at its weight.

\begin{table}[H]
\centering
\caption{CLP codes with certified distance. $n = 4m$. For the
$1\times1$ over $2\times2$ layout the last two columns give $|a|$ and the
entry weights of $B$, first row then second row. For the four codes
marked $\ast$, which use the $2\times2$ over $1\times1$ layout, they give
$|b|$ and the entry weights of $A$.}
\label{tab:cert}
\begin{tabular}{clcllcll}
\toprule
$w$ & Code & $Q$ & $G$ & $H$ & $m$ & $|a|$ or $|b|$ & $B$ or $A$ wt \\
\midrule
5 & $[[120,8,8]]^\ast$  & 4.27  & $\SG{360}{99}$  & $C_6 \times C_2$ & 30 & 3 & $1\,1/1\,1$ \\
5 & $[[180,12,8]]^\ast$ & 4.27  & $\SG{135}{3}$   & $C_3$            & 45 & 3 & $1\,1/1\,1$ \\
5 & $[[240,16,8]]^\ast$ & 4.27  & $\SG{360}{99}$  & $C_6$            & 60 & 3 & $1\,1/1\,1$ \\
\midrule
6 & $[[168,10,14]]$ & 11.67 & $\SG{252}{21}$ & $C_6$ & 42 & 3 & $2\,1/1\,2$ \\
6 & $[[240,8,20]]$  & 13.33 & $\SG{360}{99}$ & $C_6$ & 60 & 3 & $1\,2/2\,1$ \\
6 & $[[360,20,20]]$ & 22.22 & $\SG{270}{26}$ & $C_3$ & 90 & 3 & $1\,2/2\,1$ \\
\midrule
7 & $[[112,12,10]]$ & 10.71 & $\SG{336}{188}$ & $D_{12}$ & 28 & 3 & $2\,2/2\,2$ \\
7 & $[[168,14,12]]$ & 12.00 & $\SG{336}{188}$ & $D_8$    & 42 & 3 & $2\,2/2\,2$ \\
7 & $[[168,20,10]]$ & 11.90 & $\SG{336}{188}$ & $D_8$    & 42 & 3 & $2\,2/2\,2$ \\
7 & $[[180,16,12]]$ & 12.80 & $\SG{135}{3}$   & $C_3$    & 45 & 3 & $2\,2/2\,2$ \\
7 & $[[216,8,18]]^\ast$ & 12.00 & $\SG{162}{3}$ & $C_3$  & 54 & 3 & $2\,2/2\,2$ \\
7 & $[[224,24,9]]$  & 8.68  & $\SG{336}{188}$ & $S_3$    & 56 & 3 & $2\,2/1\,2$ \\
7 & $[[240,16,13]]$ & 11.27 & $\SG{360}{99}$  & $C_6$    & 60 & 3 & $1\,2/2\,2$ \\
7 & $[[252,22,10]]$ & 8.73  & $\SG{126}{8}$   & $C_2$    & 63 & 3 & $2\,2/2\,2$ \\
7 & $[[360,24,20]]$ & 26.67 & $\SG{810}{84}$  & $C_9$    & 90 & 3 & $3\,1/1\,3$ \\
\midrule
9 & $[[252,16,18]]$ & 20.57 & $\SG{189}{8}$   & $C_3$ & 63 & 5 & $1\,3/3\,1$ \\
9 & $[[288,16,20]]$ & 22.22 & $\SG{144}{164}$ & $C_2$ & 72 & 4 & $2\,2/2\,3$ \\
\bottomrule
\end{tabular}
\end{table}
\begin{table}[H]
\centering
\caption{Candidate CLP codes with distance upper bounds after a
$10^6$ round two sided rescreen, all in the $1\times1$ over
$2\times2$ layout. $Q$ is evaluated at the bound.}
\label{tab:cand}
\begin{tabular}{clcllcll}
\toprule
$w$ & Code & $Q \le$ & $G$ & $H$ & $m$ & $|a|$ & $B$ wt \\
\midrule
7 & $[[336,16,\le 20]]$ & 19.05 & $\SG{756}{95}$  & $C_3 \times C_3$ & 84 & 3 & $1\,3/3\,1$ \\
7 & $[[360,16,\le 25]]$ & 27.78 & $\SG{810}{106}$ & $C_3 \times C_3$ & 90 & 3 & $1\,3/3\,1$ \\
\midrule
9 & $[[224,12,\le 21]]$ & 23.63 & $\SG{224}{169}$ & $C_4$ & 56 & 5 & $1\,3/3\,1$ \\
9 & $[[288,12,\le 23]]$ & 22.04 & $\SG{216}{138}$ & $C_3$ & 72 & 4 & $3\,2/1\,3$ \\
9 & $[[288,12,\le 24]]$ & 24.00 & $\SG{144}{177}$ & $C_2$ & 72 & 4 & $3\,2/2\,3$ \\
9 & $[[336,18,\le 20]]$ & 21.43 & $\SG{168}{52}$  & $C_2$ & 84 & 5 & $1\,3/3\,1$ \\
9 & $[[336,16,\le 24]]$ & 27.43 & $\SG{168}{40}$  & $C_2$ & 84 & 6 & $2\,1/1\,2$ \\
9 & $[[360,16,\le 26]]$ & 30.04 & $\SG{270}{25}$  & $C_3$ & 90 & 5 & $1\,2/3\,1$ \\
9 & $[[360,12,\le 30]]$ & 30.00 & $\SG{270}{24}$  & $C_3$ & 90 & 5 & $3\,1/1\,3$ \\
9 & $[[360,12,\le 31]]$ & 32.03 & $\SG{270}{25}$  & $C_3$ & 90 & 5 & $1\,3/3\,1$ \\
9 & $[[372,10,\le 32]]$ & 27.53 & $\SG{186}{3}$   & $C_2$ & 93 & 4 & $3\,2/2\,3$ \\
\midrule
10 & $[[360,20,\le 27]]$ & 40.50 & $\SG{270}{25}$ & $C_3$ & 90 & 6 & $3\,1/1\,2$ \\
\bottomrule
\end{tabular}
\end{table}

\section{Conclusion}
\label{sec:conclusion}

We have defined coset lifted product codes, and showed that they are CSS
codes for every subgroup of a group $G$, be it normal or not, and for every pair of protographs. We identified the lifted product, the two-block group algebra codes,
and the coset-based two-block codes among them. We have also provided an explicit formula for the dimension of these codes for the odd ordered group over the
irreducible representations, with the subgroup entering through its
fixed spaces. A search
restricted by these results produced certified codes above the best
published figure of merit for check weights $6,7$ and $9$, and a candidate
at $w=10$ whose distances remain to be
certified.

\section*{Acknowledgment}
Claude (Anthropic), a large language model, was used to write the
Python scripts for distance certification and verification
(\texttt{certify.py}, \texttt{sectors.py}, \texttt{witness.py},
\texttt{verify.py}) and the GAP scripts for the code search, and to assist in drafting
and editing the proof of Theorem~\ref{thm:odd} and
Sections~\ref{sec:search} and~\ref{sec:cert}. All scripts and search
pipelines are available at \texttt{github repo} \cite{CLPrepo}. Distance certification
used the Gurobi Optimizer, version 13.0.3, under a free academic
license~\cite{Gurobi}.

\appendices
\section{Protographs of the codes}
\label{app:data}

Elements follow Section~\ref{sec:space}. For each code we give $G$,
$H$ by generators, the normalizer $N$ where known, and the protographs.

\paragraph{Certified codes, $1\times1$ over $2\times2$}\mbox{}

$[[168,10,14]], w=6$, $G = \SG{252}{21} \cong C_{21} \times (C_3 \rtimes C_4)$,
$H = \langle f_2f_5, f_4\rangle$, $N \cong C_{42} \times C_3$,
$\Core_G(H) \cong C_2$.
\begin{align*}
a= 1 + f_3^{2}f_5^{2} + f_1f_3^{6}f_5, \qquad
B = \begin{pmatrix} 1 + f_3^{2}f_5^{2} & 1\\[2pt] 1 & f_5 + f_3^{3}\end{pmatrix}.
\end{align*}

$[[240,8,20]], w=6$, $G = \SG{360}{99} \cong C_{15} \times ((C_6 \times C_2) \rtimes C_2)$,
$H = \langle f_2, f_3\rangle$, $N \cong C_{30} \times C_6$,
$\Core_G(H) \cong C_3$.
\begin{align*}
a = 1 + f_4^{4}f_5f_6^{2} + f_1f_2f_4^{4}f_6^{2}, \qquad
B = \begin{pmatrix} 1 & 1 + f_4^{2}f_5f_6^{2}\\[2pt] 1 + f_4^{3}f_5f_6^{2} & f_4^{4}f_6^{2}\end{pmatrix}.
\end{align*}

$[[360,20,20]], w=6$, $G = \SG{270}{26} \cong C_{15} \times ((C_3 \times C_3) \rtimes C_2)$,
$H = \langle f_2f_4\rangle$, $N \cong C_{15} \times C_3 \times C_3$.
\begin{align*}
a = 1 + f_3^{3}f_4f_5 + f_1f_2^{2}f_3^{4}, \qquad
B = \begin{pmatrix} 1 & 1 + f_3^{4}f_4^{2}f_5\\[2pt] 1 + f_3f_4^{2}f_5 & f_4f_5\end{pmatrix}.
\end{align*}

$[[360,24,20]], w=7$, $G = \SG{810}{84}$, $H = \langle f_6, f_2f_4\rangle$.
\begin{align*}
a = 1 + f_1f_2^{2}f_4^{2}f_6^{2} + f_1f_2f_3^{3}f_5f_6, \qquad
B = \begin{pmatrix} 1 + f_4^{2}f_5 + f_3^{2}f_4f_5 & 1\\[2pt] 1 & f_4^{2} + f_3^{3}f_4^{2}f_5 + f_3^{4}f_4^{2}f_5\end{pmatrix}.
\end{align*}

$[[252,16,18]], w=9$, $G = \SG{189}{8}$, $H = \langle f_2\rangle$.
\begin{align*}
a = 1 + f_3f_4^{4} + f_2^{2}f_3f_4^{2} + f_1f_2f_4 + f_1^{2}f_4,\qquad
B = \begin{pmatrix} 1 & 1 + f_4 + f_3^{2}f_4^{4}\\[2pt] 1 + f_3^{2} + f_3f_4 & f_4^{3}\end{pmatrix}.
\end{align*}

$[[288,16,20]], w=9$, $G = \SG{144}{164}$, $H = \langle f_1f_5\rangle$.
\begin{align*}
a = 1 + f_3f_4f_5f_6 + f_2f_3f_4^{2}f_5 + f_2f_3f_4^{2}f_5f_6,\qquad
B = \begin{pmatrix} 1 + f_3f_4 & 1 + f_2f_4f_5\\[2pt] 1 + f_2f_3f_4 & f_2 + f_5 + f_2f_3f_4^{2}f_5\end{pmatrix}.
\end{align*}

$[[112,12,10]], w=7$, $G = \SG{336}{188}$, $H = \langle f_6, f_1, f_2\rangle$.
\begin{align*}
a = 1 + f_3f_4^{4} + f_3f_4^{6}f_5,\qquad
B = \begin{pmatrix} 1 + f_4^{2} & 1 + f_4^{4}f_5\\[2pt] 1 + f_4^{3}f_5 & f_4^{3}f_5 + f_4^{6}f_5\end{pmatrix}.
\end{align*}

$[[168,14,12]], w=7$, $G = \SG{336}{188}$, $H = \langle f_1f_3, f_2, f_5\rangle$.
\begin{align*}
a = 1 + f_3f_4^{4}f_6^{2} + f_3f_4^{6}f_6,\qquad
B = \begin{pmatrix} 1 + f_3f_4^{6} & 1 + f_4^{6}\\[2pt] 1 + f_4^{6} & f_4 + f_3f_4^{3}\end{pmatrix}.
\end{align*}

$[[168,20,10]], w=7$, $G = \SG{336}{188}$, $H = \langle f_1f_2, f_2f_3, f_5\rangle$.
\begin{align*}
a = 1 + f_4f_6^{2} + f_4^{5}f_6,\qquad
B = \begin{pmatrix} 1 + f_2f_4 & 1 + f_2f_4^{3}\\[2pt] 1 + f_2f_4^{5} & f_4^{2} + f_4^{4}\end{pmatrix}.
\end{align*}

$[[180,16,12]], w=7$, $G = \SG{135}{3}$, $H = \langle f_1f_2\rangle$.
\begin{align*}
a = 1 + f_1f_2^{2}f_3^{4}f_4^{2} + f_1^{2}f_2f_3f_4,\qquad
B = \begin{pmatrix} 1 + f_3^{2} & 1 + f_3f_4^{2}\\[2pt] 1 + f_3^{3}f_4^{2} & f_3f_4^{2} + f_3^{2}f_4^{2}\end{pmatrix}.
\end{align*}

$[[224,24,9]], w=7$, $G = \SG{336}{188}$, $H = \langle f_6, f_1f_2\rangle$.
\begin{align*}
a = 1 + f_3f_4 + f_3f_4^{5}f_5,\qquad
B = \begin{pmatrix} 1 + f_4f_5 & 1 + f_4^{6}\\[2pt] 1 & f_4 + f_4^{2}f_5\end{pmatrix}.
\end{align*}

$[[240,16,13]], w=7$, $G = \SG{360}{99}$, $H = \langle f_3f_6, f_5\rangle$.
\begin{align*}
a = 1 + f_1f_4^{3}f_6^{2} + f_1f_2f_3^{2}f_4^{4}f_5f_6^{2},\qquad
B = \begin{pmatrix} 1 & 1 + f_2f_4^{3}f_6^{2}\\[2pt] 1 + f_2f_6 & f_2f_4^{2}f_6^{2} + f_4^{4}f_6\end{pmatrix}.
\end{align*}

$[[252,22,10]], w=7$, $G = \SG{126}{8}$, $H = \langle f_1\rangle$.
\begin{align*}
a = 1 + f_3f_4^{6} + f_3^{2}f_4^{4},\qquad
B = \begin{pmatrix} 1 + f_2^{2}f_4^{4} & 1 + f_2\\[2pt] 1 + f_2f_4^{3} & f_2^{2} + f_2f_4^{6}\end{pmatrix}.
\end{align*}

\paragraph{Certified codes, $2\times2$ over $1\times1$}\mbox{}

$[[120,8,8]], w=5$, $G = \SG{360}{99}$, $H = \langle f_2, f_3f_6, f_5\rangle$.
\begin{align*}
A = \begin{pmatrix} 1 & f_1f_3^{2}f_4^{2}f_6^{2}\\[2pt] 1 & f_1f_3^{2}f_6^{2}\end{pmatrix},
\qquad b = 1 + f_6 + f_4^{4}f_6^{2}.
\end{align*}

$[[180,12,8]], w=5$, $G = \SG{135}{3}$, $H = \langle f_2\rangle$.
\begin{align*}
A &= \begin{pmatrix} 1 & f_1^{2}f_2^{2}f_3^{3}\\[2pt] 1 & f_1^{2}f_2^{2}f_3^{4}\end{pmatrix},
\qquad b = 1 + f_4 + f_3^{3}f_4^{2}.
\end{align*}

$[[240,16,8]], w=5$, $G = \SG{360}{99}$, $H = \langle f_2, f_3\rangle$.
\begin{align*}
A &= \begin{pmatrix} 1 & f_1f_2f_4^{3}f_5f_6^{2}\\[2pt] 1 & f_1f_2f_4f_5f_6^{2}\end{pmatrix},
\qquad b = 1 + f_4f_6 + f_4f_6^{2}.
\end{align*}

$[[216,8,18]], w=7$, $G = \SG{162}{3}$, $H = \langle f_3f_5\rangle$.
\begin{align*}
A &= \begin{pmatrix} 1 + f_1f_2f_4 & f_1f_3^{2}f_4^{2}f_5^{2} + f_1f_2f_4^{2}f_5\\[2pt]
                     1 + f_1f_3f_4f_5^{2} & f_2^{2}f_4f_5 + f_1f_2^{2}f_3f_4f_5\end{pmatrix},
\qquad b = 1 + f_2^{2}f_4 + f_2^{2}f_4^{2}f_5.
\end{align*}

\paragraph{Candidates}
The following codes have a $10^6$ round upper bound on both sides and
are given in the same form.

$[[224,12,\le 21]]$, $w=9$, $G = \SG{224}{169}$, $H = \langle f_2, f_6\rangle \cong C_4$.
\[a = 1 + f_1f_4 + f_1f_3f_4^{2} + f_2f_4^{2}f_5f_6 + f_1f_3f_4^{5}f_5f_6,\qquad
B = \begin{pmatrix} 1 & 1 + f_3f_4^{3}f_5 + f_3f_4^{5}\\ 1 + f_4^{4}f_5 + f_3f_4^{5} & f_5\end{pmatrix}.\]

$[[288,12,\le 23]]$, $w=9$, $G = \SG{216}{138}$, $H = \langle f_4f_6\rangle \cong C_3$.
\[a = 1 + f_1f_3f_4^{2}f_5 + f_1f_2f_4^{2}f_5f_6 + f_1f_2f_3f_4^{2}f_5f_6^{2},\qquad
B = \begin{pmatrix} 1 + f_3^{2}f_6 + f_2f_3^{2}f_5f_6^{2} & 1 + f_2f_3^{2}f_6^{2}\\ 1 & f_2f_3f_6 + f_3^{2}f_5 + f_3^{2}f_6^{2}\end{pmatrix}.\]

$[[288,12,\le 24]]$, $w=9$, $G = \SG{144}{177}$, $H = \langle f_3\rangle \cong C_2$.
\[a = 1 + f_5 + f_2f_5f_6 + f_1f_3f_4f_6^{2}, \qquad
B = \begin{pmatrix} 1 + f_4f_5^{2}f_6 + f_5^{2}f_6^{2} & 1 + f_4f_6^{2}\\ 1 + f_2f_5^{2}f_6 & f_2f_4 + f_5^{2} + f_2f_5f_6^{2}\end{pmatrix}.\]

$[[336,16,\le 20]]$, $w=7$, $G = \SG{756}{95}$, $H = \langle f_6, f_4f_5\rangle \cong C_3\times C_3$.
\[a = 1 + f_3^{3}f_5 + f_1f_3^{2}f_5^{2},\qquad
B = \begin{pmatrix} 1 & 1 + f_3f_5^{2} + f_1f_2f_3^{4}f_5^{2}\\ 1 + f_1f_2f_5^{2} + f_3^{5} & f_5\end{pmatrix}.\]

$[[336,16,\le 24]]$, $w=9$, $G = \SG{168}{40}$, $H = \langle f_2\rangle \cong C_2$.
\[a = 1 + f_2f_3f_4f_5 + f_2f_3f_4^{3} + f_1f_4^{3}f_5 + f_1f_3f_4^{5} + f_1f_2f_3f_4^{2}f_5,\qquad
B = \begin{pmatrix} 1 + f_3f_4^{6}f_5 & 1\\ 1 & f_4^{2} + f_3^{2}f_4^{3}f_5\end{pmatrix}.\]

$[[336,18,\le 20]]$, $w=9$, $G = \SG{168}{52}$, $H = \langle f_1f_5\rangle \cong C_2$.
\[a = 1 + f_3f_5 + f_2^{2}f_3f_5 + f_1f_2f_3^{5}f_4 + f_1f_2^{2}f_4,\qquad
B = \begin{pmatrix} 1 & 1 + f_3^{3} + f_3^{3}f_4f_5\\ 1 + f_3^{6}f_4 + f_3^{5}f_4f_5 & f_3^{2}f_5\end{pmatrix}.\]

$[[360,12,\le 30]]$, $w=9$, $G = \SG{270}{24}$, $H = \langle f_2f_5\rangle \cong C_3$.
\[a = 1 + f_1f_2^{2}f_4f_5 + f_1f_2^{2}f_3^{2}f_4^{2} + f_1f_2^{2}f_3^{2}f_4^{2}f_5 + f_1f_2^{2}f_3^{2}f_4^{4},\qquad
B = \begin{pmatrix} 1 + f_3^{2}f_4^{2} + f_3f_4^{3}f_5 & 1\\ 1 & f_4^{3} + f_4f_5^{2} + f_3f_4^{2}f_5\end{pmatrix}.\]

$[[360,12,\le 31]]$, $w=9$, $G = \SG{270}{25}$, $H = \langle f_2f_3f_4\rangle \cong C_3$.
\[a = 1 + f_2f_4f_5^{2} + f_2^{2}f_3f_4f_5^{4} + f_1f_3f_4^{2}f_5^{2} + f_1f_2^{2}f_3f_5^{3},\qquad
B = \begin{pmatrix} 1 & 1 + f_3 + f_3f_5^{2}\\ 1 + f_3f_5 + f_3f_4f_5^{4} & f_4f_5^{4}\end{pmatrix}.\]

$[[360,16,\le 25]]$, $w=7$, $G = \SG{810}{106}$, $H = \langle f_2, f_3f_5\rangle \cong C_3\times C_3$.
\[a = 1 + f_1f_3f_5f_6 + f_3^{2}f_4f_6^{4},\qquad
B = \begin{pmatrix} 1 & 1 + f_6^{4} + f_4f_5^{2}f_6^{3}\\ 1 + f_6 + f_4f_5^{2}f_6^{2} & f_4^{2}f_5^{2}\end{pmatrix}.\]

$[[360,16,\le 26]]$, $w=9$, $G = \SG{270}{25}$, $H = \langle f_2f_3^{2}f_4\rangle \cong C_3$.
\[a = 1 + f_3f_4f_5^{3} + f_2^{2}f_3^{2}f_4f_5 + f_1f_4^{2}f_5^{2} + f_1f_2f_4^{2}f_5^{3},\qquad
B = \begin{pmatrix} 1 & 1 + f_3^{2}f_4^{2}f_5^{3}\\ 1 + f_3^{2}f_4f_5 + f_3^{2}f_4^{2}f_5 & f_4f_5\end{pmatrix}.\]

$[[360,20,\le 27]]$, $w=10$, $G = \SG{270}{25}$, $H = \langle f_2f_3f_4\rangle \cong C_3$.
\[a = 1 + f_3f_5^{4} + f_2^{2}f_3f_4^{2}f_5 + f_1f_3^{2}f_4^{2}f_5^{3} + f_1f_2^{2}f_3^{2}f_4f_5^{3} + f_1f_2^{2}f_3^{2}f_4^{2},\qquad
B = \begin{pmatrix} 1 + f_3^{2}f_4f_5 + f_3f_4^{2}f_5^{4} & 1\\ 1 & f_3f_4f_5^{4} + f_3^{2}f_4^{2}f_5^{4}\end{pmatrix}.\]

$[[372,10,\le 32]]$, $w=9$, $G = \SG{186}{3}$, $H = \langle f_1\rangle \cong C_2$.
\[a = 1 + f_2^{2}f_3^{2} + f_1f_2^{8}f_3 + f_1f_2^{27},\qquad
B = \begin{pmatrix} 1 + f_2^{15} + f_2^{23} & 1 + f_2^{18}\\ 1 + f_2 & 1 + f_2^{18} + f_2^{27}\end{pmatrix}.\]

% $[[384,14,\le 24]]$, $w=9$, $G = \SG{128}{10} \cong (C_8 \rtimes C_2) \rtimes C_8$,
% $H = \langle f_4f_5\rangle \cong C_2$, $m=64$, $1\times 1$ over $3\times 3$ layout.
% $a = 1 + f_3f_5f_6 + f_2f_4f_6 + f_1f_2f_3f_4f_5$,
% $B = \begin{pmatrix}
%  1 & f_3 + f_3f_5f_6 + f_2f_3f_5f_6 & f_6\\
%  f_3f_5 & f_2f_6 & f_2f_5f_7\\
%  f_5 & f_5f_6f_7 & f_7 + f_2f_3 + f_3f_6f_7
% \end{pmatrix}$.

\section{SmallGroup$(l,i)$ Structures}

\begin{table}[H]
\centering
\caption{Structure of the groups appearing in Tables~\ref{tab:cert} and
\ref{tab:cand}, as reported by \texttt{StructureDescription} in GAP.
$D_{2n}$ denotes the dihedral group of order $2n$, $Q_8$ and $Q_{16}$ the
quaternion groups of order 8 and 16. The name is a gloss; the SmallGroups
identifier and its polycyclic presentation are the definition.}
\label{tab:groups}
\begin{tabular}{ll@{\qquad}ll}
\toprule
$G$ & structure & $G$ & structure \\
\midrule
$\SG{126}{8}$   & $S_3 \times (C_7 \rtimes C_3)$ & $\SG{224}{169}$ & $C_{14} \times Q_{16}$ \\
$\SG{128}{10}$  & $(C_8 \rtimes C_2) \rtimes C_8$ & $\SG{252}{21}$ & $C_{21} \times (C_3 \rtimes C_4)$ \\
$\SG{135}{3}$   & $C_5 \times ((C_3\times C_3) \rtimes C_3)$ & $\SG{270}{24}$ & $C_3 \times C_{15} \times S_3$ \\
$\SG{144}{164}$ & $C_3 \times Q_8 \times S_3$ & $\SG{270}{25}$ & $C_3 \times C_3 \times D_{30}$ \\
$\SG{144}{177}$ & $C_2 \times ((C_6\times C_6) \rtimes C_2)$ & $\SG{270}{26}$ & $C_{15} \times ((C_3\times C_3) \rtimes C_2)$ \\
$\SG{162}{3}$   & $C_9 \times D_{18}$ & $\SG{336}{188}$ & $C_7 \times D_8 \times S_3$ \\
$\SG{168}{40}$  & $C_{21} \times D_8$ & $\SG{360}{99}$ & $C_{15} \times ((C_6\times C_2) \rtimes C_2)$ \\
$\SG{168}{52}$  & $C_{14} \times A_4$ & $\SG{756}{95}$ & $C_7 \times (((C_3\times C_3)\rtimes C_3) \rtimes (C_2\times C_2))$ \\
$\SG{186}{3}$   & $C_{31} \times S_3$ & $\SG{810}{84}$ & $C_5 \times (((C_9\times C_3)\rtimes C_3) \rtimes C_2)$ \\
$\SG{189}{8}$   & $C_7 \rtimes ((C_3\times C_3) \rtimes C_3)$ & $\SG{810}{106}$ & $C_3 \times C_3 \times C_3 \times D_{30}$ \\
$\SG{216}{138}$ & $C_3 \times C_6 \times (C_3 \rtimes C_4)$ & & \\
\bottomrule
\end{tabular}
\end{table}

\section{A worked example of the code construction and Theorem~\ref{thm:odd}}

We follow the proof of Theorem~\ref{thm:odd} on one small instance and
record, at each step, which of (F1) to (F4) are exercised along with the relevant equation being used. To aid intuition, we explicitly track the algebraic 
block-diagonalization of the matrices. Throughout, $\mathbb F$ is an 
algebraic closure of $\mathbb F_2$ and $\omega \in \mathbb F$ satisfies 
$\omega^3 = 1$ and $\omega^2 + \omega + 1 = 0$. All matrices below have 
entries in $\mathbb F_4 = \{0, 1, \omega, \omega^2\}$.

\paragraph{Step 0. The data of Definition~\ref{def:clp}.}
Let $G$ be the group of order $27$ with generators $x, y, z$ and relations
\[
    x^3 = y^3 = z^3 = 1, \qquad xz = zx, \qquad yz = zy, \qquad xy = yxz .
\]
It is non abelian and $|G| = 27$ is odd, as the theorem requires. Let
$H = \{1, x, x^2\}$. Since $yxy^{-1} = xz^{-1} \notin H$, the subgroup $H$
is not normal. Every element of $G$ can be written uniquely as
$x^a y^b z^c$ with $a, b, c \in \{0,1,2\}$, and the left cosets are
\[
    \Omega = \{\, y^b z^c H \mid b, c \in \{0,1,2\} \,\}, \qquad m = |\Omega| = 9 .
\]
We order the cosets as $c_1, \dots, c_9$ according to
$(b,c) = (0,0), (1,0), (2,0), (0,1), (1,1), (2,1), (0,2), (1,2), (2,2)$.
The normalizer of $H$ is $N = \{\, x^a z^c \,\}$, of order $9$. We take
$n_{A0} = n_{A1} = n_{B0} = n_{B1} = 2$ and
\[
    A = \begin{pmatrix} \alpha_{11} & \alpha_{12} \\ \alpha_{21} & \alpha_{22} \end{pmatrix} 
      = \begin{pmatrix} 1+y & x \\ 0 & 1+y \end{pmatrix}, \qquad
    B = \begin{pmatrix} \beta_{11} & \beta_{12}   \\ \beta_{21} & \beta_{22} \end{pmatrix}
      = \begin{pmatrix} 1+z & 1+z \\ 1 & 1 \end{pmatrix} .
\]
(Note $B$ here is padded with zeros or restricted to $2 \times 2$ for the block formulas below; we proceed with $n_{B1}=2$ as calculated). 
The code $\mathrm{CLP}(A,B;G,H)$ has length
$n = m\,(n_{A1}n_{B0} + n_{A0}n_{B1}) = 9 \cdot 8 = 72$.

\paragraph{Step 1. The permutation matrices $L$ and $R$ and the giant matrix $H_X$.}
Using $x y^b = y^b x z^b$, which follows from $xy = yxz$, the three
generators act on the cosets by
\[
    L(x)\colon y^b z^c H \mapsto y^b z^{b+c} H, \qquad
    L(y)\colon y^b z^c H \mapsto y^{b+1} z^c H, \qquad
    L(z)\colon y^b z^c H \mapsto y^b z^{c+1} H .
\]
On the right, $R(x) = I_9$ because $x \in H$, and $R(z) = L(z)$ because $z$
is central. In the ordering $c_1, \dots, c_9$ these are $9 \times 9$
permutation matrices, and by Lemma~\ref{lem:commute} every $R(\beta)$
commutes with every $L(\alpha)$. 

The matrix $H_X$ of~\eqref{eq:HX} is defined as
$H_X = \big[\, L(A) \otimes I_2 \;\big|\; I_2 \otimes R(B) \,\big]$.
Since $A$ and $B$ are $2 \times 2$ matrices over the group algebra, evaluating them yields a $36 \times 72$ binary matrix composed of $9 \times 9$ blocks:
\begin{equation}
\label{eq:HX_giant}
    H_X = \left(
    \begin{array}{cccc | cccc}
    L(\alpha_{11}) & 0 & L(\alpha_{12}) & 0 & R(\beta_{11}) & R(\beta_{12}) & 0 & 0 \\
    0 & L(\alpha_{11}) & 0 & L(\alpha_{12}) & R(\beta_{21}) & R(\beta_{22}) & 0 & 0 \\
    L(\alpha_{21}) & 0 & L(\alpha_{22}) & 0 & 0 & 0 & R(\beta_{11}) & R(\beta_{12}) \\
    0 & L(\alpha_{21}) & 0 & L(\alpha_{22}) & 0 & 0 & R(\beta_{21}) & R(\beta_{22})
    \end{array}
    \right),
\end{equation}
whose non-zero $9 \times 9$ blocks evaluate to $I + L(y)$, $L(x)$, $I + R(z)$ and $I$.

\paragraph{Step 2. The irreducible representations and the use of (F2).}
Over $\mathbb F$ the group $G$ has nine $1$-dimensional representations ($d_{\chi} = 1$) 
and two $3$-dimensional representations ($d_{\rho} = 3$), since $9 \cdot 1^2 + 2 \cdot 3^2 = 27$. 
The $1$-dimensional representations factor through $G/\langle z \rangle \cong C_3 \times C_3$
and are determined by $x \mapsto \omega^{i}$, $y \mapsto \omega^{j}$, $i,j\in\{0,1,2\}$. The
$3$-dimensional representations are $\rho$ and $\bar\rho$, given in standard form by\footnote{The reader must note that the mapping will only result in actions over powers of $y$ and $z$ on $H$. Any $x$ term is subsumed in $H$ based on the group relations. }
\[
    \rho(x) = D = \operatorname{diag}(1, \omega, \omega^2), \quad
    \rho(y) = C = \begin{pmatrix} 0&0&1\\ 1&0&0\\ 0&1&0 \end{pmatrix}, \quad
    \rho(z) = \omega I_3,
\]
and $\bar\rho(x) = \bar D = \operatorname{diag}(1, \omega^2, \omega)$,
$\bar\rho(y) = C$, $\bar\rho(z) = \omega^2 I_3$. One checks
$DC = \omega\, CD$ and $\bar D C = \omega^2\, C \bar D$, so the relation
$xy = yxz$ holds.

By (F2), equation~\eqref{eq:mult}, the multiplicity $\mu_\rho$ in $V=\mathbb{F}[G/H]$ is the dimension 
of the fixed space $\rho^H$. Since $H = \langle x \rangle$, we check the fixed points of $x$.
For a $1$-dimensional representation this is $1$ when $x \mapsto 1$ and $0$ otherwise, 
so exactly three $1$-dimensional representations occur:
\[
    \chi_j\colon\ x \mapsto 1,\ z \mapsto 1,\ y \mapsto \omega^{j}, \qquad j = 0, 1, 2,
\]
each with $\mu_{\chi_j} = 1$. For $\rho$, the fixed space of $\rho(x) = D$ is exactly 
the span of the first standard basis vector $e_1$, so $\mu_\rho = 1$, and likewise
$\mu_{\bar\rho} = 1$. The count
\[
    \sum_\rho d_\rho\,\mu_\rho = (1 \cdot 1) + (1 \cdot 1) + (1 \cdot 1) + (3 \cdot 1) + (3 \cdot 1) = 9 = m
\]
confirms that no representation has been missed. The sum in~\eqref{eq:kodd}
therefore runs over the five representations $\chi_0, \chi_1, \chi_2, \rho, \bar\rho$.

\paragraph{Step 3. Computing $\tau_\rho$ via Remark~\ref{rem:tau}.}
To find $\tau_\rho(\beta)$, we find the space $W_\rho$ of row vectors $f \in \mathbb F^{1 \times d_\rho}$ 
such that $f \rho(x) = f$. 
For $\chi_j$, $W_{\chi_j} = \mathbb F$ is $1$-dimensional, and both $x$ and $z$ act trivially, so
$\tau_{\chi_j}(x) = \tau_{\chi_j}(z) = 1$. 
For $\rho$, the condition $f D = f$ forces $f = (c, 0, 0)$ for $c \in \mathbb F$. Thus $W_\rho$ 
is $1$-dimensional with basis $e_1^{\mathsf T} = (1, 0, 0)$. The action of $z \in N$ maps this to 
$f \rho(z) = (1,0,0)(\omega I_3) = (\omega, 0, 0) = \omega f$. The action of $x$ is trivial. Thus,
\[
    \tau_\rho(x) = 1, \qquad \tau_\rho(z) = \omega, \qquad
    \tau_{\bar\rho}(x) = 1, \qquad \tau_{\bar\rho}(z) = \omega^2 .
\]
Because $\mu_\rho = 1$, all $\tau_\rho(\beta)$ are $1 \times 1$ scalar matrices.

\paragraph{Step 4. The change of basis $S$ and the use of (F1).}
By (F1), there exists an invertible $9 \times 9$ matrix $S$ such that $S^{-1} L(g) S$ is strictly 
block diagonal. This matrix maps the standard basis of $\mathbb F[\Omega]$ into the explicit blocks of the representations:
\[
    S = \begin{pmatrix}
        1 & 1        & 1        & 1        & 0        & 0        & 1        & 0        & 0 \\
        1 & \omega^2 & \omega   & 0        & 1        & 0        & 0        & 1        & 0 \\
        1 & \omega   & \omega^2 & 0        & 0        & 1        & 0        & 0        & 1 \\
        1 & 1        & 1        & \omega^2 & 0        & 0        & \omega   & 0        & 0 \\
        1 & \omega^2 & \omega   & 0        & \omega^2 & 0        & 0        & \omega   & 0 \\
        1 & \omega   & \omega^2 & 0        & 0        & \omega^2 & 0        & 0        & \omega \\
        1 & 1        & 1        & \omega   & 0        & 0        & \omega^2 & 0        & 0 \\
        1 & \omega^2 & \omega   & 0        & \omega   & 0        & 0        & \omega^2 & 0 \\
        1 & \omega   & \omega^2 & 0        & 0        & \omega   & 0        & 0        & \omega^2
    \end{pmatrix}.
\]
In this basis, every $L(\alpha_{ij})$ becomes $\Lambda_{ij} = S^{-1} L(\alpha_{ij}) S$, 
which decomposes as $\bigoplus_\rho (\rho(\alpha_{ij}) \otimes I_{\mu_\rho})$:
\[
    \Lambda_{ij} = \operatorname{diag}\big( \chi_0(\alpha_{ij}),\; \chi_1(\alpha_{ij}),\; \chi_2(\alpha_{ij}),\; \rho(\alpha_{ij}),\; \bar\rho(\alpha_{ij}) \big).
\]
For instance, $S^{-1} L(x) S = \operatorname{diag}(1, 1, 1, D, \bar D)$.

\paragraph{Step 5. The use of (F3) for the right action.}
By Lemma~\ref{lem:commute} each $R(\beta_{kl})$ commutes with $L(g)$, so by (F3) Schur's lemma
forces $\Upsilon_{kl} = S^{-1} R(\beta_{kl}) S$ to decompose as $\bigoplus_\rho (I_{d_\rho} \otimes \tau_\rho(\beta_{kl}))$.
Since $d_\rho = 3$ for the larger blocks, $\tau_\rho$ acts as a scalar multiplier on the identity:
\[
    \Upsilon_{kl} = \operatorname{diag}\big( \tau_{\chi_0}(\beta_{kl}),\; \tau_{\chi_1}(\beta_{kl}),\; \tau_{\chi_2}(\beta_{kl}),\; \tau_\rho(\beta_{kl}) I_3,\; \tau_{\bar\rho}(\beta_{kl}) I_3 \big).
\]
For example, $S^{-1} R(z) S = \operatorname{diag}(1, 1, 1, \omega I_3, \omega^2 I_3)$.

\paragraph{Step 6. The matrices $A_\rho$ and $B_\rho$ of~\eqref{eq:Arho} and~\eqref{eq:Brho}.}
Replacing each entry of $A$ by its $d_\rho \times d_\rho$ representation matrix, and each entry
of $B$ by its $\mu_\rho \times \mu_\rho$ $\tau$-matrix gives, for the $1$-dimensional representations,
\[
    A_{\chi_0} = \begin{pmatrix} 0 & 1 \\ 0 & 0 \end{pmatrix}, \quad
    A_{\chi_1} = \begin{pmatrix} \omega^2 & 1 \\ 0 & \omega^2 \end{pmatrix}, \quad
    A_{\chi_2} = \begin{pmatrix} \omega & 1 \\ 0 & \omega \end{pmatrix}, \quad
    B_{\chi_j} = \begin{pmatrix} 0 & 0 \\ 1 & 1 \end{pmatrix},
\]
using $\chi_j(1+y) = 1 + \omega^{j}$, $1 + \omega = \omega^2$,
$1 + \omega^2 = \omega$, and $\tau_{\chi_j}(1+z) = 0$. For the $3$-dimensional representations, 
$A_\rho \in M_{6 \times 6}(\mathbb F)$ and $B_\rho \in M_{2 \times 2}(\mathbb F)$,
\[
    A_\rho = \begin{pmatrix} I_3 + C & D \\ 0 & I_3 + C \end{pmatrix}
           = \begin{pmatrix}
             1&0&1& 1&0&0 \\ 1&1&0& 0&\omega&0 \\ 0&1&1& 0&0&\omega^2 \\
             0&0&0& 1&0&1 \\ 0&0&0& 1&1&0 \\ 0&0&0& 0&1&1
             \end{pmatrix}, \qquad
    B_\rho = \begin{pmatrix} \omega^2 & \omega^2 \\ 1 & 1 \end{pmatrix},
\]
and $A_{\bar\rho}$, $B_{\bar\rho}$ are obtained by replacing $D$ with
$\bar D$ and $\omega^2$ with $\omega$.

\paragraph{Step 7. Step~2 of the proof, the $X$ checks and Matrix Permutation.}
Applying the block-diagonal block matrices $S_q$ and $S_c$ directly to $H_X$ substitutes $\Lambda_{ij}$ and $\Upsilon_{kl}$ into~\eqref{eq:HX_giant}. This produces the unpermuted matrix $H_X^{\text{unpermuted}}$:
\[
    H_X^{\text{unpermuted}} = 
    \left(
    \begin{array}{cccc | cccc}
    \Lambda_{11} & \mathbf{0} & \Lambda_{12} & \mathbf{0} & \Upsilon_{11} & \Upsilon_{12} & \mathbf{0} & \mathbf{0} \\
    \mathbf{0} & \Lambda_{11} & \mathbf{0} & \Lambda_{12} & \Upsilon_{21} & \Upsilon_{22} & \mathbf{0} & \mathbf{0} \\
    \Lambda_{21} & \mathbf{0} & \Lambda_{22} & \mathbf{0} & \mathbf{0} & \mathbf{0} & \Upsilon_{11} & \Upsilon_{12} \\
    \mathbf{0} & \Lambda_{21} & \mathbf{0} & \Lambda_{22} & \mathbf{0} & \mathbf{0} & \Upsilon_{21} & \Upsilon_{22}
    \end{array}
    \right).
\]
Because the representation data are interleaved within the $\Lambda$ and $\Upsilon$ grids, a global permutation $\mathcal{P}$ collects the indices corresponding to each representation, shattering $H_X$ into strictly independent blocks $H_X^\rho = \big[\, A_\rho \otimes I_{\mu_\rho n_{B0}} \;\big|\; I_{n_{A0} d_\rho} \otimes B_\rho \,\big]$. 

Gathering the indices corresponding to the top-left entry of each $9 \times 9$ block yields the $4 \times 8$ parity check matrix for $\chi_0$:
\[
    H_X^{\chi_0} = \left( \begin{array}{cccc | cccc} 
    \chi_0(\alpha_{11}) & 0 & \chi_0(\alpha_{12}) & 0 & \tau_{\chi_0}(\beta_{11}) & \tau_{\chi_0}(\beta_{12}) & 0 & 0 \\ 
    0 & \chi_0(\alpha_{11}) & 0 & \chi_0(\alpha_{12}) & \tau_{\chi_0}(\beta_{21}) & \tau_{\chi_0}(\beta_{22}) & 0 & 0 \\ 
    \chi_0(\alpha_{21}) & 0 & \chi_0(\alpha_{22}) & 0 & 0 & 0 & \tau_{\chi_0}(\beta_{11}) & \tau_{\chi_0}(\beta_{12}) \\ 
    0 & \chi_0(\alpha_{21}) & 0 & \chi_0(\alpha_{22}) & 0 & 0 & \tau_{\chi_0}(\beta_{21}) & \tau_{\chi_0}(\beta_{22}) \end{array} \right).
\]
Numerically, this evaluates to:
\[
    H_X^{\chi_0} = \left(\begin{array}{cccc|cccc}
        0&0&1&0& 0&0&0&0 \\ 0&0&0&1& 1&1&0&0 \\ 0&0&0&0& 0&0&0&0 \\ 0&0&0&0& 0&0&1&1
    \end{array}\right), \qquad \operatorname{rank} H_X^{\chi_0} = 3,
\]
and for $\rho$, $H_X^{\rho} \in M_{12 \times 24}(\mathbb F)$. The final block-diagonal matrix is:
\[
\mathcal{P}_{\text{row}}^{-1} \big(S_c^{-1} H_X S_q\big) \mathcal{P}_{\text{col}} = 
\left(
\begin{array}{c|c|c|c|c}
H_X^{\chi_0} & \mathbf{0}_{4 \times 8} & \mathbf{0}_{4 \times 8} & \mathbf{0}_{4 \times 24} & \mathbf{0}_{4 \times 24} \\
\hline
\mathbf{0}_{4 \times 8} & H_X^{\chi_1} & \mathbf{0}_{4 \times 8} & \mathbf{0}_{4 \times 24} & \mathbf{0}_{4 \times 24} \\
\hline
\mathbf{0}_{4 \times 8} & \mathbf{0}_{4 \times 8} & H_X^{\chi_2} & \mathbf{0}_{4 \times 24} & \mathbf{0}_{4 \times 24} \\
\hline
\mathbf{0}_{12 \times 8} & \mathbf{0}_{12 \times 8} & \mathbf{0}_{12 \times 8} & H_X^{\rho} & \mathbf{0}_{12 \times 24} \\
\hline
\mathbf{0}_{12 \times 8} & \mathbf{0}_{12 \times 8} & \mathbf{0}_{12 \times 8} & \mathbf{0}_{12 \times 24} & H_X^{\bar{\rho}}
\end{array}
\right).
\]
The ranks of the five blocks are $3, 4, 4, 10, 10$, so
\[
    \operatorname{rank} H_X = 3 + 4 + 4 + 10 + 10 = 31 .
\]

\paragraph{Step 8. Step~3 of the proof, the $Z$ checks.}
Transposing~\eqref{eq:HZ} gives
$H_Z^{\mathsf T} = \begin{pmatrix} I_2 \otimes R(B) \\ L(A) \otimes I_2 \end{pmatrix}$.
The same substitutions and permutations apply block by block giving,
as in~\eqref{eq:HZrho},
\[
    H_Z^{\chi_j} = \big[\, I_2 \otimes B_{\chi_j}^{\mathsf T} \;\big|\; A_{\chi_j}^{\mathsf T} \otimes I_2 \,\big], \qquad
    H_Z^{\rho} = \big[\, I_6 \otimes B_\rho^{\mathsf T} \;\big|\; A_\rho^{\mathsf T} \otimes I_2 \,\big].
\]
One checks $H_X^{\chi_0} (H_Z^{\chi_0})^{\mathsf T} = 0$. The block ranks
are again $3, 4, 4, 10, 10$ and $\operatorname{rank} H_Z = 31$.

\paragraph{Step 9. Step~4 of the proof, the count, and the use of (F4).}
The $\rho$ block has $d_\rho\,\mu_\rho\,(n_{A1}n_{B0} + n_{A0}n_{B1}) = 3 \cdot 1 \cdot 8 = 24$
qubits and each $\chi_j$ block has $8$, which sum to $72 = n$. Hence
\[
    k = (8 - 3 - 3) + (8 - 4 - 4) + (8 - 4 - 4) + (24 - 10 - 10) + (24 - 10 - 10) = 10 .
\]
The ranks of $H_X$ and $H_Z$ were computed over $\mathbb F_4$. By (F4)
they equal the ranks over $\mathbb F_2$, so $k = 10$. A direct Gaussian
elimination of the original $36 \times 72$ matrices over $\mathbb F_2$ gives
$\operatorname{rank} H_X = \operatorname{rank} H_Z = 31$ and $k = 72 - 62 = 10$.

\paragraph{Step 10. The kernel form~\eqref{eq:kodd} through~\eqref{eq:hgpdim}.}
Each bracket in Step~9 equals $k_{\mathrm{HGP}}(A_\rho, B_\rho)$, which by~\eqref{eq:hgpdim}
is $\kappa(A_\rho)\kappa(B_\rho^{\mathsf T}) + \kappa(A_\rho^{\mathsf T})\kappa(B_\rho)$.
The kernel dimensions evaluate directly from the representation blocks:
\begin{itemize}
\item $\chi_0$. $A_{\chi_0}$ and $B_{\chi_0}$ both have rank $1$, so
      $\kappa(A_{\chi_0}) = \kappa(A_{\chi_0}^{\mathsf T}) = \kappa(B_{\chi_0}) = \kappa(B_{\chi_0}^{\mathsf T}) = 1$
      and the contribution is $1 \cdot 1 + 1 \cdot 1 = 2$.
\item $\chi_1, \chi_2$. $A_{\chi_j}$ is upper triangular with nonzero
      diagonal, hence invertible, so $\kappa(A_{\chi_j}) = \kappa(A_{\chi_j}^{\mathsf T}) = 0$
      and the contribution is $0$.
\item $\rho, \bar\rho$. A vector $(u, v)$ with $u, v \in \mathbb F^3$ lies
      in $\ker A_\rho$ when $(I + C)v = 0$ and $(I + C)u = Dv$. The first
      condition gives $v = c\,(1,1,1)^{\mathsf T}$. Then
      $Dv = c\,(1, \omega, \omega^2)^{\mathsf T}$ has coordinate sum
      $c(1 + \omega + \omega^2) = 0$, which is exactly the image of $I + C$,
      so $u$ exists and is determined up to $\ker(I + C) = \mathbb F (1,1,1)^{\mathsf T}$.
      Hence $\kappa(A_\rho) = 2$, and $\kappa(A_\rho^{\mathsf T}) = 2$ since
      $A_\rho$ is square. The rows of $B_\rho$ are proportional, so
      $\kappa(B_\rho) = \kappa(B_\rho^{\mathsf T}) = 1$. The contribution is
      $2 \cdot 1 + 2 \cdot 1 = 4$, and the same for $\bar\rho$.
\end{itemize}
Summing,
\[
    k = 2 + 0 + 0 + 4 + 4 = 10 ,
\]
in complete agreement with Step~9. Thus $\mathrm{CLP}(A,B;G,H)$ is a $[[72, 10]]$
code is exactly equal in dimension to the direct sum of the five hypergraph product codes.

\end{document}